\documentclass[a4paper,UKenglish,cleveref, autoref, thm-restate]{lipics-v2021}

\hideLIPIcs  
\nolinenumbers 
\usepackage{graphicx,tikz}
\usepackage{todonotes}
\graphicspath{{./graphics/}}

\usepackage{xspace}

\newcommand{\defparaproblem}[4]{
 \vspace{2mm}
\noindent\fbox{
 \begin{minipage}{0.96\textwidth}
 \begin{tabular*}{\textwidth}{@{\extracolsep{\fill}}lr} #1 & \\ \end{tabular*}
 {\textbf{Input:}} #2 \\
 {\textbf{Parameter:}} #3 \\
 {\textbf{Question:}} #4
 \end{minipage}
 }
 \vspace{2mm}
}

\DeclareMathOperator{\cw}{\bf cw}
\DeclareMathOperator{\lcw}{\bf lcw}
\DeclareMathOperator{\pw}{\bf pw}
\DeclareMathOperator{\logpw}{\bf logpw}
\DeclareMathOperator{\loglcw}{\bf loglcw}

\newcommand{\N}{\mathbb{N}}
\newcommand{\Z}{\mathbb{Z}}

\newcommand\mimval{\mathrm{mim}}
\newcommand\cutmim{\mathrm{cutmim}}
\newcommand\mimw{\mathrm{mimw}}
\newcommand\cutw{\mathrm{cutw}}
\newcommand\card[1]{|#1|}
\newcommand\Splitcomp{Subdivision-complement\xspace}
\newcommand\splitcomp{subdivision-complement\xspace}
\newcommand\CMC{CMC\xspace}

\newcommand\CMClique{\textsc{Chained Multicolored Clique}\xspace}
\newcommand\CMISet{\textsc{Chained Multicolored Independent Set}\xspace}
\newcommand\ISet{\textsc{Independent Set}\xspace}
\newcommand\DomSet{\textsc{Dominating Set}\xspace}
\newcommand\qCol{\textsc{$q$-Coloring}\xspace}

\newcommand\fiveLCol{\textsc{$5$-List-Coloring}\xspace}
\newcommand\fiveCol{\textsc{$5$-Coloring}\xspace}
\newcommand\FVSet{\textsc{Feedback Vertex Set}\xspace}
\newcommand\MIForest{\textsc{Maximum Induced Forest}\xspace}
\newcommand\calI{\mathcal{I}}
\newcommand\calP{\mathcal{P}}

\title{XNLP-completeness for Parameterized Problems on Graphs with a Linear Structure}
\author{Hans L. Bodlaender}{Utrecht University, The Netherlands}{h.l.bodlaender@uu.nl}{https://orcid.org/0000-0002-9297-3330}{
}
\author{Carla Groenland}{Utrecht University, The Netherlands}{c.e.groenland@uu.nl}{http://orcid.org/0000-0002-9878-8750}{Supported by the project CRACKNP that has received funding from the European Research Council (ERC) under the European Union’s Horizon 2020 research innovation programme (grant agreement No 853234).}
\author{Hugo Jacob}{ENS Paris-Saclay, France}{hugo.jacob@ens-paris-saclay.fr}{https://orcid.org/0000-0003-1350-3240}{}
\author{Lars Jaffke}{University of Bergen, Norway}{lars.jaffke@uib.no}{https://orcid.org/0000-0003-4856-5863}{Supported by the Norwegian Research Council (project number 274526) and the Meltzer Research Fund.}
\author{Paloma T.\ Lima}{IT University of Copenhagen, Denmark}{palt@itu.dk}{https://orcid.org/
0000-0001-9304-4536}{}
\authorrunning{H.\,L. Bodlaender, C. Groenland, H. Jacob, L. Jaffke, and P.\,T. Lima} 
\Copyright{Hans L. Bodlaender, Carla Groenland, Hugo Jacob, Lars Jaffke, and Paloma T. Lima} 

\EventEditors{}
\EventNoEds{0}
\EventLongTitle{}
\EventShortTitle{}
\EventAcronym{}
\EventYear{}
\EventDate{}
\EventLocation{}
\EventLogo{}
\SeriesVolume{}
\ArticleNo{}

\begin{document}

\maketitle

\begin{abstract}
    In this paper, we showcase the class XNLP as a natural place for many hard problems parameterized by linear width measures. This strengthens existing W[1]-hardness proofs for these problems, since XNLP-hardness implies $W[t]$-hardness for all $t$. It also indicates, via a conjecture by Pilipczuk and Wrochna [ToCT 2018], that any XP algorithm for such problems is likely to require XP space. 
    
    In particular, we show XNLP-completeness for natural problems parameterized by pathwidth, linear clique-width, and linear mim-width. The problems we consider are \textsc{Independent Set}, \textsc{Dominating Set}, \textsc{Odd Cycle Transversal}, \textsc{($q$-)Coloring}, \textsc{Max Cut}, \textsc{Maximum Regular Induced Subgraph}, \textsc{Feedback Vertex Set}, \textsc{Capacitated (Red-Blue) Dominating Set}, and \textsc{Bipartite Bandwidth}.
\end{abstract}

\section{Introduction}
Since the inception of parameterized complexity in the late 1980s and early 1990s, much research has
been done on establishing the complexity of parameterized problems. 
Typically one is particularly interested in either designing FPT-algorithms for these problems, or to prove them $W[t]$-hard, for some $t$, which provides evidence that such a problem is not likely to be fixed-parameter tractable. 
As opposed to the classical P versus NP-complete setting, the question of membership in some class of the $W$-hierarchy is often much less clear.
While some natural problems such as \textsc{Independent Set} and \textsc{Dominating Set} are known to be $W[1]$-complete and $W[2]$-complete, respectively, many other problems are unknown to be complete
for a class of parameterized problems, and even conjectured not to be in the $W$-hierarchy.
%
Recently, building upon work by Elberfeld et al.~\cite{ElberfeldST15}, Bodlaender et al.~\cite{XNLP-comp} 
introduced a complexity class called XNLP, which gives a way of addressing this question.


The class XNLP consists of the parameterized problems that can be solved with a non-deterministic algorithm that uses $f(k)\log n$ space and $f(k)n^c$ time, where $f$ is a computable function, $n$ is the input size, $k$ is the parameter and $c$ is a constant. In particular, XNLP-hardness implies $W[t]$-hardness for all $t$. Therefore it is unlikely that any XNLP-hard problem is complete for some $W[t]$.

One success story within parameterized algorithms and complexity is the use of width measures of graphs as parameters (see, e.g.,~\cite{CyganFKLMPPS15}).
Typically, such width measures are defined in terms of a tree-like decomposition of a graph,
and the width describes the complexity of the decomposition, and therefore, in turn, of the graph.
Such width measures also have linear variants, where the decomposition resembles a path instead of a tree.
In this work, we provide evidence that the class 
\emph{XNLP is the `natural home' for hard
problems parameterized by linear width measures.}

Let us give some intuitive explanation why this is the case.
A typical dynamic programming algorithm that uses such a linear decomposition 
stores, at each node of the path, some partial solutions associated with it.
The table entries associated with the nodes are then filled in the order in which they appear on the path.
If one turns such an algorithm into a nondeterministic algorithm, 
it often suffices at the $i$-th node to nondeterministically determine the 
table index corresponding to the correct partial solution (if it exists) from the table entry that was 
previously determined for the $(i-1)$-th node.
In such a case, membership in XNLP follows if 
each single table entry of such a DP algorithm can be represented by $f(k)\log n$ bits 
(where $k$ is the width)
and if the nondeterministic step does not require a computation that uses significantly more space.
This is often the case.
Now, such an approach fails for tree-like decompositions, 
since even a nondeterministic algorithm might have to keep 
too many table entries at some point during the computation.
One common situation in which this occurs is when the algorithm needs to store one table entry 
for each level of the decomposition. 
This incurs a multiplicative factor in the memory usage that depends on the height of the tree,
which can be prohibitively large.

In this direction, Bodlaender et al.~\cite{XNLP-comp} showed that \textsc{List Colouring} parameterized by the pathwidth of the input graph, and \textsc{Bandwidth} are XNLP-complete.
In this paper, 
we show XNLP-completeness of fundamental graph problems parameterized by linear variants of well-established width measures, such as pathwidth, linear clique-width and linear mim-width, as well as some of their logarithmic analogues. 
 


Besides showing $W[t]$-hardness for all $t$, XNLP-hardness also provides insight into the space complexity of parameterized problems. Pilipczuk and Wrochna~\cite{PilipczukW18} proposed the following 
conjecture.\footnote{The statement of the conjecture here is
equivalent to the conjecture on time and memory use for the {\sc Longest Common Subsequence} problem
from \cite{PilipczukW18}; the name of the conjecture is taken as analogue to
the naming of XP as problems that use \emph{slice-wise polynomial time} (see \cite[Section 1.1]{CyganFKLMPPS15}).}

\begin{conjecture}[Slice-wise Polynomial Space Conjecture~\cite{PilipczukW18}]
\label{conj}
XNLP-hard problems do not have an algorithm, that runs in $n^{f(k)}$ time and $f(k)n^c$ space, with $f$ a computable function, $k$ the parameter, $n$ the input size, and $c$ a constant.
\end{conjecture}


Typically, membership in XP for the problems studied in our paper follows from a dynamic programming approach that uses a significant amount of memory.  XNLP-hardness indicates (via Conjecture~\ref{conj}) that dynamic programming is in some sense `optimal' (no XP algorithm can use `significantly less' memory).

\subparagraph*{Linear width measures and logarithmic analogues.}
The width measures we consider in this work include linear variants of arguably the most prominent measures, and some of their generalizations.
Pathwidth is a linear variant of the classic treewidth parameter, which, informally speaking, measures how close a connected graph is to being a tree.
In this vein, pathwidth measures how close a connected graph is to being a path.
Clique-width (or, equivalently, rank-width) generalizes treewidth to several simply structured dense graphs,
and its linear counterpart is called linear clique-width (linear rank-width).
Maximum induced matching width~\cite{VatshelleThesis}, or mim-width for short, 
in turn generalizes clique-width and remains bounded even on well-studied graph classes such as interval and permutation graphs, 
where the clique-width is known to be unbounded.
In fact, for most of these classes the \emph{linear} mim-width is bounded.
%
%

We also introduce a new parameter that we call logarithmic linear clique-width, analogous to the parameter logarithmic pathwidth that was introduced by Bodlaender et al.~\cite{XNLP-comp}. For an $n$-vertex graph of linear clique-width $k$, logarithmic linear clique-width takes the value $\lceil k/\log n\rceil$. 
We stress the fact that XNLP-hardness parameterized by a logarithmic parameter implies that there is no algorithm solving the problem in time $2^{O(k)}n^{O(1)}$ and space $k^{O(1)}n^{O(1)}$, where $k$ is the original parameter\footnote{Indeed, replacing $k$ with $k'\log n$, this gives running time $2^{O(k'\log n)}=n^{O(k')}$ and space $k'^{O(1)}n^{O(1)}$, which is excluded by the conjecture.}, under Conjecture \ref{conj}. Such results can complement existing (S)ETH lower bounds for single exponential FPT algorithms with lower bounds on the space requirements of such algorithms.


\subparagraph*{Bipartite bandwidth.}
Finally, we consider a bipartite variant of the notoriously difficult~\cite{Bodlaender21} problem 
of computing the bandwidth of a graph.
Here, for a bipartite graph with vertex bipartition $(A, B)$, and bandwidth target value $w$, 
we want to find an ordering $\alpha$ of $A$ and an ordering $\beta$ of $B$,
such that for each edge $ab$, $|\alpha(a) - \beta(b)| \le w$.
We consider this problem parameterized by $w$, 
and show that it is XNLP-complete, even when the input graph is a tree.

\subparagraph*{Our results.}
We summarize our results in the following theorem.
\begin{theorem}\label{thm:main}
    The following problems are XNLP-complete.
\begin{enumerate}[(i)]
     \item\label{thm:main:pw} 
    {\sc Capacitated Red-Blue Dominating Set} and {\sc Capacitated Dominating Set} parameterized by pathwidth.
    \item\label{thm:main:lcw} 
    \textsc{Coloring}, \textsc{Maximum Regular Induced Subgraph}, and {\sc Max Cut} parameterized by linear clique-width.
    \item\label{thm:main:log}
    {\sc $q$-Coloring} and  {\sc Odd Cycle Transversal}
    parameterized by logarithmic pathwidth or logarithmic linear clique-width.
    \item\label{thm:main:mim}
    \textsc{Independent Set}, \textsc{Dominating Set}, \textsc{Feedback Vertex Set}, and \textsc{$q$-Coloring} for fixed $q \ge 5$ parameterized by linear mim-width.
    \item\label{thm:main:bandwidth} 
    {\sc Bipartite Bandwidth}, even if the input graph is a tree.
\end{enumerate}
    Furthermore, \textsc{Feedback Vertex Set} parameterized by logarithmic pathwidth or logarithmic linear clique-width is XNLP-hard.
\end{theorem}

Note that \cref{thm:main}(\ref{thm:main:lcw}) and~(\ref{thm:main:mim}) include the first XNLP-completeness results for graph problems with the linear 
clique-width and linear mim-width as parameter.


\subparagraph*{Related Work.}
Guillemot~\cite{Guillemot11} introduced the class WNL (which equals XNLP closed under fpt-reductions), and
showed some problems to be complete for WNL, including a version of {\sc Longest Common Subsequence}. The
class XNLP (under a different name) was introduced by Elberfeld et al.~\cite{ElberfeldST15}, who
also showed a number of problems, including {\sc Linear Cellular Automaton Acceptance}, to be complete for the
class. A large number of parameterized problems was shown to be XNLP-complete recently by Bodlaender
et al.~\cite{XNLP-comp}. Very recently, in work that aims at separating the complexity of
treewidth and pathwidth at one side, and stable gonality at another side, 
Bodlaender et al.~\cite{BodlaenderCW22} showed a number of flow problems parameterized by pathwidth to be complete for XNLP. 

\section{Overview of the results}\label{sec:overview}
In this section, we give a bird's-eye view of the results proved in this paper,
and discuss related work for the specific problems we consider.

\subparagraph*{Parameterized by linear clique-width.}
We consider the \textsc{Max Cut}, the \textsc{Coloring}, and the \textsc{Maximum Regular Induced Subgraph} problem
parameterized by linear clique-width.
For \textsc{Max Cut}, 
let $E(V_1,V_2)$ denote the set of edges with one endpoint in $V_1$ and one endpoint in $V_2$. 

\defparaproblem{\sc Max Cut}{A graph $G=(V,E)$ described by a given linear $k$-expression describing $G$ and an integer $W$.}{$k$.}{Is there a bipartition of $V$ into $(V_1,V_2)$ such that $|E(V_1,V_2)| \geq W$?}

In 1994, Wanke~\cite{Wanke94} showed that {\sc Max Cut} is in XP for
graphs of bounded NLC-width, which directly implies XP-membership 
with clique-width as parameter, as NLC-width and clique-width are linearly
related. In 2014,
Fomin et al.~\cite{FominGLS14} consider the fine grained complexity for
\textsc{Max Cut} for graphs of small clique-width, giving an algorithm
with improved running time and showing asymptotic optimality (assuming the Exponential Time Hypothesis). 
From their results, it follows
that \textsc{Max Cut} is $W[1]$-hard with clique-width as parameter.
In \cref{sec:maxcut}, we prove the following theorem.
\begin{restatable}{theorem}{maxcut} 
{\sc Max Cut} with linear clique-width as parameter is XNLP-complete.
\end{restatable}

Next, we consider the classical \textsc{Coloring} problem, 
which given a graph $G$ and an integer $k$ asks if $G$ has a proper coloring with $k$ colors.
Similarly to the story of the \textsc{Max Cut} problem, 
\textsc{Coloring} parameterized by clique-width was shown to be in XP by Wanke in 1994~\cite{Wanke94},
and a W[1]-hardness proof only followed in 2010 by Fomin et al.~\cite{FominGLS10}.
The XP algorithm for coloring runs in time $n^{O(2^k)}$, where $k$ is the clique-width,
and Fomin et al.~\cite{FominGLSZ19} 
even showed that this run time can probably not be substantially improved:
an algorithm running in time $n^{2^{o(k)}}$ would refute the ETH.
We prove the following in \cref{sec:col}.

\begin{restatable}{theorem}{coloringcw}
    \textsc{Coloring} parameterized by linear clique-width is XNLP-complete.
\end{restatable}

Lastly, we consider the \textsc{Maximum Regular Induced Subgraph}
problem. 
The problem was studied by several authors, including
Asahiro et al.~\cite{AsahiroEIM14}, who show among others an algorithm
that uses linear time for graphs of bounded treewidth, where the time
depends single exponentially on the treewidth. 
Moser and Thilikos~\cite{MoserT09}, 
and independently Mathieson and Szeider~\cite{MathiesonS08} show (amongst other results) that the problem is $W[1]$-hard when the size of the subgraph (parameter $W$ in our description below) is used
as parameter.
Broersma et al.~\cite{BroersmaGP13} give XP algorithms for
several problems, including \textsc{Maximum Regular Induced Subgraph}
for graphs of bounded clique-width.
The proof of the theorem below is given in \cref{sec:maxris}.

\defparaproblem{\textsc{Maximum Regular Induced Subgraph}}{A graph $G$ described by a given linear $k$-expression and two integers $W$ and $D$.}{$k$.}{Is there a $D$-regular induced subgraph of $G$ on at least $W$ vertices?}

\begin{restatable}{theorem}{regularsubgraph}
    \label{thm:regular-induced-subgraph}
    \textsc{Maximum Regular Induced Subgraph} parameterized by linear clique-width is XNLP-complete.
\end{restatable}

\subparagraph*{Parameterized by pathwidth.}
We consider the \textsc{Capacitated Red-Blue Dominating Set} and
{\sc Capacitated Dominating Set} problems. Below, we give the formal statement of the
problems, where we have the width of the path decomposition as parameter.
One of the reasons of interest in these problems is that they model facility location problems: the red vertices
model possible facilities that can serve a bounded number of clients which are modelled by the blue vertices.

\defparaproblem{\textsc{Capacitated Red-Blue Dominating Set}}{A bipartite graph $G=(R,B,E)$, a path decomposition of $G$ of width $\ell$, a capacity function $c:R \to \N$, and an integer $k$.}{$\ell$.}{Is there a subset $S$ of $R$, and an assignment of blue vertices $f:B \to S$ such that $\{w,f(w)\}\in E$ for all $w\in R$ and $|f^{-1}(v)| \leq c(v)$ for all $v$ in $S$?}

\defparaproblem{\textsc{Capacitated Dominating Set}}{A graph $G=(V,E)$, a path decomposition of $G$ of width $\ell$, a capacity function $c:V \to \N$, and an integer $k$.}{$\ell$.}{Is there a subset $S$ of $V$, and an assignment of the vertices $f:V \to S$ such that $\{w,f(w)\}\in E$ or $w=f(w)$ for all $w\in V$ and $|f^{-1}(v)| \leq c(v)$ for all $v$ in $S$?}

In 2008, Dom et al.~\cite{DomLSV08} showed that \textsc{Capacitated Dominating Set} is
$W[1]$-hard, with the treewidth and solution size $k$ as combined parameter.
\textsc{Capacitated Dominating Set} was shown to be $W[1]$-hard for planar graphs,
with the solution size as parameter by Bodlaender et al.~\cite{BodlaenderLP09}.
Fomin et al.~\cite{FominGLS14} give bounds for the fine grained complexity of
\textsc{Capacitated Red-Blue Dominating Set}, for graphs with a small feedback vertex
set; their results imply that the problem is $W[1]$-hard with feedback vertex set as parameter.
The proof of the following theorem can be found in \cref{sec:rbds}.
\begin{restatable}{theorem}{cds}
\label{thm:cds}
{\sc Capacitated Red-Blue Dominating Set} and  {\sc Capacitated Dominating Set} parameterized by pathwidth are XNLP-complete.
\end{restatable}

\subparagraph*{Parameterized by logarithmic linear clique-width.}
Bodlaender et al.~\cite{XNLP-comp} introduced the parameter \textit{logarithmic pathwidth} as $\pw/\log n$ for an $n$-vertex graph of pathwidth $\pw$. This allows the pathwidth to be linear in the logarithm of the number of vertices of
the graph. Here we introduce the \emph{logarithmic linear clique-width} as $\lcw/\log n$ for graphs on $n$ vertices with linear clique-width $\lcw$.

We provide new XNLP-complete problems for the parameter logarithmic pathwidth, and show that these problems and the previously known XNLP-complete problems for this parameter \cite{XNLP-comp} are also complete for the parameter logarithmic linear clique-width. Our results are summarised below.

The motivation to study the logarithmic linear clique-width or logarithmic 
pathwidth comes from the observation that many FPT algorithms with
linear cliquewidth or  pathwidth as parameter have a single exponential
time dependency on the parameter. Thus, if linear cliquewidth or
pathwidth is logarithmic in the size of the graph, these algorithms turn into
XP algorithms.

\begin{restatable}{theorem}{logwidth}
\label{thm:logwidth}
    When parameterized by logarithmic pathwidth or logarithmic linear clique-width, \textsc{Independent Set}, \textsc{Dominating Set}, $q$-\textsc{List-Coloring} for $q>2$, and \textsc{Odd Cycle Transversal} are XNLP-complete, and \textsc{Feedback Vertex Set} is XNLP-hard.  
\end{restatable}

Lokshtanov et al.~\cite{LokshtanovMS18known} established (tight) lower bounds for these problems for the parameter pathwidth under the Strong Exponential Time Hypothesis. Several of our gadgets are based on those used for these lower bounds by \cite{LokshtanovMS18known}. 
We give the problem definitions and the proof of the theorem in \cref{sec:loglcw}.

\subparagraph*{Parameterized by linear mim-width.}
We prove that several fundamental graph problems are XNLP-complete
when parameterized by the mim-width of a given linear order of the input graph.
W[1]-hardness for \textsc{Independent Set} and \textsc{Dominating Set} 
in this parameterization was shown by Fomin et al.~\cite{FGR20},
and for \textsc{Feedback Vertex Set} by Jaffke et al.~\cite{JKT20}.
For \textsc{$q$-Coloring}, $W[1]$-hardness was not known before our work.
We would like to point out that our XNLP-hardness proof uses a gadget that requires five colors to construct,
and it would be interesting to see if this can be improved to three colors.
The following theorem is proved in \cref{sec:mim}.

\begin{restatable}{theorem}{mimwidth}
\label{thm:mim}
    The following problems are XNLP-complete parameterized by the mim-width of a given linear order of the vertices of the input graph:
    \begin{enumerate}
        \item\label{thm:mim:IS} \textsc{Independent Set}
        \item\label{thm:mim:DS} \textsc{Dominating Set}
        \item\label{thm:mim:qCol} \textsc{$q$-Coloring} for any fixed $q \ge 5$
        \item\label{thm:mim:FVS} \textsc{Feedback Vertex Set}
    \end{enumerate}
\end{restatable}

\subparagraph*{Bipartite bandwidth.}
We consider the following bipartite variant of the
\textsc{Bandwidth} problem.

\defparaproblem{\textsc{Bipartite Bandwidth}}{A bipartite graph $G=(X,Y,E)$ and an integer $k$.}{$k$.}{Are there orderings $\alpha:X \to [n]$ and $\beta:Y\to [m]$ such that for each $uv \in E$, $|\alpha(u)-\beta(v)| \leq k$ ?}

A possible application of this problem is as follows. Let a matrix $M$ be given. Create a vertex $x_i\in X$ for each row $i$ and a vertex $y_j\in Y$ for each column $j$, and let $x_i$ be adjacent to $y_j$ if and only if $M_{i,j}\neq 0$. This graph has bipartite bandwidth at most $k$ if and only if the rows and columns of $M$ can be permuted (individually) in such a way that all non-zero entries are within $k$ distance from the main diagonal. The following result is proved in \cref{sec:bw}.

\begin{restatable}{theorem}{bandwidth}
\label{thm:bandwidth}
{\sc Bipartite Bandwidth} is XNLP-complete for trees.
\end{restatable}

\section{Preliminaries}

The required background on the computational problems studied in this paper are given in their respective sections. The notions relevant to the entire paper are defined below.

We write $[n]=\{1,\dots,n\}$ and $[a,b]$ for the set of integers $x$ with $a\leq x\leq b$.
All logarithms in this paper have base $2$. We use $\N$ for the set of the natural numbers $\{0, 1, 2, \ldots \}$, and $\Z^+$ denotes the set of the positive natural numbers $\{1,2, \ldots\}$.
We write $N(S)$ and $N[S]=N(S)\cup S$ for the open and closed neighborhood of $S$.

\subsection{Definition of the class XNLP}
\label{subsec:defXNLP}
In this paper, we study parameterized decision problems, which are subsets of $\Sigma^\ast \times \N$, for a finite alphabet $\Sigma$. We assume the reader to be familiar with notions from parameterized complexity, such as XP, $W[1]$,
$W[2]$, \ldots, $W[P]$ (see e.g.  \cite{DowneyF99}).

The class XNLP (denoted $N[f~ \text{poly}, f \log]$ by \cite{ElberfeldST15}) consists of the parameterized decision problems that can be solved by a non-deterministic algorithm that simultaneously uses at most $f(k) n^c$ time and at most 
$f(k) \log n$ space, on an input $(x,k)$, where $x$ can be denoted with $n$ bits, $f$ a computable function, and $c$ a constant. We assume that functions $f$ of the parameter in time and resource bounds are computable --- this is called {\em strongly uniform} by
Downey and Fellows~\cite{DowneyF99}.

More information about the complexity class XNLP can be found in \cite{XNLP-comp}.

\subsection{Reductions}
\label{subsec:reductions}
In the remainder of the paper, unless stated otherwise, completeness for XNLP is with respect
to pl-reductions, which are defined below. The definitions
are based upon the formulations in \cite{ElberfeldST15}. 
\begin{itemize}
    \item A {\em parameterized reduction} from a parameterized problem $Q_1 \subseteq \Sigma_1^\ast \times \N$ to a parameterized problem $Q_2 \subseteq \Sigma_2^\ast \times \N$ is a function
    $f : \Sigma_1^\ast \times \N \rightarrow \Sigma_2^\ast \times \N$, such that the following holds.
    \begin{enumerate}
        \item For all $(x,k) \in \Sigma_1^{\ast} \times \N$, $(x,k)\in Q_1$ if and only if $f((x,k)) \in Q_2$.
        \item There is a computable function $g$, such that for all $(x,k) \in \Sigma_1^\ast \times \N$, if $f((x,k)) = (y,k')$, then $k' \leq g(k)$.
    \end{enumerate}
    \item A {\em parameterized logspace reduction} or {\em pl-reduction} is a parameterized reduction
    for which there is an algorithm that computes $f((x,k))$ in space $O(g(k) + \log n)$, with $g$ a computable
    function and $n=|x|$ the number of bits to denote $x$. 
\end{itemize}

\subsection{Pathwidth, linear clique-width, and linear mim-width}
A \textit{path decomposition} of a graph $G=(V,E)$ is a sequence $(X_1, X_2, \ldots, X_r)$ of subsets of $V$  with 
the following properties.
\begin{enumerate}
    \item $\bigcup_{1\leq i\leq r} X_i = V$.
    \item For all $\{v,w\}\in E$, there is an $i\in I$ with $v,w\in X_i$.
    \item For all $1\leq i_0<i_1<i_2 \leq r$, $X_{i_0}\cap X_{i_2} \subseteq X_{i_1}$.
\end{enumerate}
The {\em width} of a path decomposition $(X_1, X_2, \ldots, X_r)$ equals $\max_{1\leq i\leq r} |X_i|-1$, and the {\em pathwidth} $\pw$ of a graph $G$ is the minimum width of a path decomposition of $G$.

A \emph{$k$-labeled} graph is a graph $G=(V,E)$ together with a labeling function $\Gamma : V(G) \to [k]$.
A \emph{$k$-expression} constructs a $k$-labeled graph by the means of the following operations:
\begin{enumerate}
    \item \emph{Vertex creation}: $i(v)$ is the $k$-labeled graph consisting of a single vertex $v$ which is assigned label $i$.
    \item \emph{Disjoint union}: $H \oplus G$ is the disjoint union of $k$-labeled graphs $H$ and $G$
    \item \emph{Join}: $\eta_{i \times j}(G)$ is the $k$-labeled graph obtained by adding all possible edges between vertices with label $i$ and vertices with label $j$ to $G$.
    \item \emph{Renaming label}: $\rho_{i \to j}(G)$ is the $k$-labeled graph obtained by assigning label $j$ to all vertices labelled $i$ in $G$.
\end{enumerate}
A \emph{linear} $k$-expression is a $k$-expression with the additional condition that one of the arguments of the disjoint union operation needs to be a graph consisting of a single vertex.
The \emph{clique-width} $\cw(G)$(resp. \emph{linear clique-width} $\lcw(G)$) of a graph $G$ is the minimal $k$ such that $G$ can be constructed by a $k$-expression (resp. linear $k$-expression) with any labeling.

For a graph $G = (V, E)$ and $A, B \subseteq V$ with $A \cap B = \emptyset$, 
we let $G[A, B]$ be the bipartite subgraph of $G$ with vertices $A \cup B$ and 
edges $\{ab \mid ab \in E, a \in A, b \in B\}$.
We let $\cutmim_G(A, B)$ be the size of a maximum induced matching in $G[A, B]$ and 
$\mimval_G(A) = \cutmim_G(A, V \setminus A)$.
Here, an induced matching $M \subseteq E$ is a matching such that 
there are no additional edges between the endpoints of $M$ in the graph in question.
The \emph{mim-width} of a linear order $v_1, \ldots, v_n$ of $V$ is the maximum, over all $i$, of 
$\mimval_G(\{v_1, \ldots, v_i\})$.
The \emph{linear mim-width} of $G$ is the minimum mim-width over all linear orders of $V$.

\subsection{Chained variants of \textsc{Satisfiability} and \textsc{Multicolored Clique}}

In \cite{XNLP-comp}, the following problems were introduced, and shown to be XNLP-complete.

\defparaproblem{{\sc Chained Positive CNF-SAT}}{$r$ sets of Boolean variables $X_1, X_2, \ldots X_r$, each of size $q$; an integer $k\in \N$; Boolean formula $\phi$, which is in conjunctive normal form and an expression on $2q$ variables, using only positive literals; for each $i$,
a partition of $X_i$ into $X_{i,1}, \ldots, X_{i,k}$ such that $\forall j,j' \in [k],|X_{i,j}|=|X_{i,j'}|$. }{$k$.}{Is it possible to satisfy the formula 
\[ \bigwedge_{1\leq i\leq r-1} \phi(X_i,X_{i+1}) \]
by setting from each set $X_{i,j}$ exactly $1$ variable to true and all others to false?}

\defparaproblem
    {\textsc{Chained Multicolored Clique}}
    {Graph $G$, partition of $V(G)$ into $V_1, \ldots, V_r$, such that for each edge $uv \in E(G)$ with $u \in V_i$ and $v \in V_j$, $|i-j| \le 1$, function $f \colon V(G) \to [k]$.}
    {$k$.}
    {Is there a set $W \subseteq V(G)$ such that for all $i \in [r-1]$, $W \cap (V_i \cup V_{i+1})$ is a clique, and for each $i \in [r]$ and $j \in [k]$, there is a vertex $v \in W \cap V_i$ with $f(v) = j$?}
    
The \textsc{Chained Multicolored Independent Set} problem is defined analogously, with the only difference that the solution $W$ is required to be an independent set.
\begin{theorem}[Bodlaender et al.~\cite{XNLP-comp}]\label{thm:chained:MCC:MCI}
\textsc{Chained Positive CNF-SAT},  \textsc{Chained Multicolored Clique} and \textsc{Chained Multicolored Independent Set} are XNLP-complete.
\end{theorem}
\section{Problems parameterized by linear clique-width}\label{sec:lcw}
In this section we prove XNLP-completeness for several problems 
parameterized by linear clique-width.
In \cref{sec:maxcut}, we consider the \textsc{Max Cut} problem,
in \cref{sec:col} the \textsc{Coloring} problem, and
in \cref{sec:maxris}, the \textsc{Maximum Induced Regular Subgraph} problem.

\subsection{Max Cut}
\label{sec:maxcut}
In this section, we consider the \textsc{Max Cut} problem, with the linear clique-width
as parameter, and show it to be XNLP-complete.
Our result is based upon the XNLP-hardness result for a problem, called
\textsc{Circulating Orientation}, with pathwidth as parameter. Borrowing from
terminology from flows in graphs, we say that a directed graph $G=(V,A)$ with for each
arc $a\in A$ a weight $w(a)\in {\bf N}$, is a \emph{circulation}, if for each
vertex $v$, the total weight of all incoming arcs at $v$ equals the total weight of
all arcs outgoing from $v$. We reduce from the following problem.

\defparaproblem{\textsc{Circulating Orientation}}{An undirected graph $G=(V,E)$ with
a path decomposition of $G$ of width $\ell$, an edge weight function $w: E \rightarrow \N$, given in unary notation.}{$\ell$.}{Is there an orientation of $G$ that is a circulation?}

\begin{theorem}[Bodlaender et al.~\cite{BodlaenderCW22}]
{\sc Circulating Orientation} is XNLP-complete.
\end{theorem}

\maxcut*
\begin{proof}
We first show membership in XNLP.
The main idea is to turn the existing dynamic programming that solves the problem given a $k$-expression of
an $n$-vertex graph of linear clique-width $k$ into a non-deterministic algorithm, by guessing an element from a table
instead of building full tables. 
For each vertex creation, we guess on which side of the partition the vertex is. We maintain the following certificate: for each label, the number of vertices on each side of the bipartition, and the number of edges of the current expression that were in the cut. Since there are at most $k$ labels and the size of the cut is bounded by the number of edges, this certificate uses only $O(k\log n)$ bits.

To show hardness for XNLP, we reduce from {\sc Circulating Orientation} with pathwidth as parameter.
Suppose we have an instance for {\sc Circulating Orientation}: an undirected
graph $G$ with edge weight function $w$.
For each vertex $v$, write $D(v)$ as the total weight of all edges incident to $v$.

We build a new, undirected graph $H=(V_H,E_H)$ as follows. For each vertex $v\in V$, each
edge $e$ with $v$ as one of its endpoints, and each integer $i \in [1,w(e)]$, we create a
vertex $x_{v,e,i}$. Two distinct vertices $x_{v,e,i}$ and $x_{w,e',j}$ are adjacent if and only
if $v=w$ or $e=e'$. In other words: for each vertex $v\in V$, we have a clique with $D(v)$ vertices, which
consists of all vertices of the form $x_{v,\cdot,\cdot}$, that we call {\em the clique of $v$}.
For each edge $e = \{v,w\} \in E$, we have a clique with $2 w(e)$ vertices,
namely all vertices of the form $x_{v,e,\cdot}$ and $x_{w,e,\cdot}$. See Figure~\ref{figure:edgecut} for a
partial example.

\begin{figure}
    \centering
    \includegraphics[scale=2]{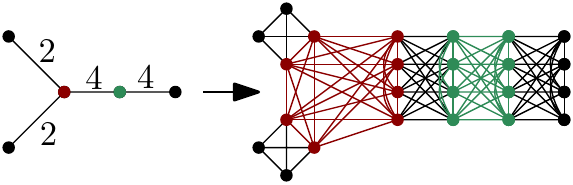}
    \caption{Example for the construction of the hardness proof of {\sc Max Cut} (fragment).}
    \label{figure:edgecut}
\end{figure}

\begin{claim}
$G$ has a circulating orientation if and only if $H$ has a bipartition that cuts
$ \sum_{e\in E} w(e)^2 + \sum_{v\in V} D(v)^2/4$ edges.
\end{claim}

\begin{claimproof}
Suppose $G$ has a circulating orientation. For each edge $e=\{v,w\}$, if the orientation directs
$v$ to $w$, then add all vertices of the form $x_{v,e,i}$ to $Z_1$ and all vertices of the form
$x_{w,e,i}$ to $Z_2$ ($i\in [1,w(e)]$); otherwise, add all vertices of the form $x_{v,e,i}$ to $Z_2$ and all vertices of the form
$x_{w,e,i}$ to $Z_1$ ($i\in [1,w(e)]$). 

Since we started from a circulating orientation, for each vertex $v$ there are $D(v)/2 \times D(v)/2$ edges of the form $\{x_{v,\cdot,\cdot},x_{v,\cdot,\cdot}\}$ crossing the bipartition. Moreover, there are $w(e)\times w(e)$ edges of the form $\{x_{v,e,\cdot},x_{w,e,\cdot}\}$ crossing the bipartition for each edge $e=\{v,w\}$. We conclude that the bipartition cuts the required number of edges.

\smallskip

Now, suppose we have a partition $Z_1$, $Z_2$ of $V_H$ with $ \sum_{e\in E} w(e)^2 + \sum_{v\in V} D(v)^2/4$ edges between $Z_1$ and $Z_2$.
We distinguish two types of edges in $E_H \cap (Z_1 \times Z_2)$. A Type 1 edge is an edge between
two vertices $x_{v,e,i}$ and $x_{v,e',j}$ (i.e., it is in the clique of a vertex $v$). A Type 2 edge is an edge between two vertices
$x_{v,e,i}$ and $e_{w,e,j}$ for some edge $e= \{v,w\}$. Note that each edge in $H$ is of Type 1 or Type 2 and  that $H$ has precisely $ \sum_{e\in E} w(e)^2$ Type 2 edges.

For each vertex $v\in V$, we consider how many Type 1 edges (those in the clique of $v$) are in $Z_1\times Z_2$.
If we have $\alpha$ vertices in the clique of $v$ that belong to $Z_1$, then $D(v)-\alpha$ vertices
in the clique of $v$ belong to $Z_2$, and thus, in this clique, we cut $\alpha \cdot (D(v)-\alpha)
\leq D(v)^2/4$ edges; the maximum possible is reached when $\alpha=D(v)/2$.

It follows that the number of Type 1 edges that are cut is at most $\sum_{v\in V} D(v)^2/4$.
So, we must cut all Type 2 edges, i.e., for each edge $e=\{v,w\}$, all edges
of the form $\{x_{v,e,i}, x_{w,e,j}\}$ are between a vertex in $Z_1$ and a vertex in $Z_2$. It follows
that we either have that all vertices of the form $x_{v,e,i}$ are in $Z_1$ and all vertices of the
form $x_{w,e,i}$ are in $Z_2$ --- in which case we direct the edge $e$ from $v$ to $w$; or
all vertices of the form $x_{v,e,i}$ are in $Z_2$ and all vertices of the
form $x_{w,e,i}$ are in $Z_1$, and now we direct the edge from $w$ to $v$.

For each vertex $v\in V$, we must have exactly $D(v)/2$ vertices from the clique of $v$ in $Z_1$ and
equally many vertices in $Z_2$; otherwise, we cannot reach the required number of cut edges.
Now, the total weight of all edges that we directed out of $v$ precisely equals the number of
vertices in the clique of $v$ in $Z_1$, and similarly, the total weight of all edges that we directed towards of $v$ precisely equals the number of
vertices in the clique of $v$ in $Z_2$. Both numbers equal $D(v)/2$. As this holds for each $v\in V$,
the orientation defined above is a circulation.
\end{claimproof}

Finally, we show that we can construct a linear clique expression for $H$
given a path decomposition of $G$; the number of colors we use for the clique width construction 
equals the width of the path decomposition plus 4. The construction uses ideas for constructing
clique width constructions for line graphs of graphs of bounded treewidth;
see \cite{GurskiW07a}.

Suppose we have a nice path decomposition, which uses introduce vertex, introduce edge, and forget nodes.
We use $k+1$ active colors --- each active color will correspond to one vertex in the current bag.
We also have an inactive color, which we will denote by the letter $o$. We also use two
temporary colors, which we call $\alpha$ and $\beta$.

We sequentially visit the bags of the path decomposition, in order. Bags correspond to a number
of steps of the construction of $H$, as described below.

If we introduce a vertex, we select a currently unused active color, and say this is the color of that
vertex, and assume it to be used. 

If we introduce an edge $e=\{v,w\}$, we add the vertices $x_{v,e,i}$ one by one, each with the
color $\alpha$. Then, we add the vertices $x_{w,e,i}$ one by one, each with the color $\beta$.
Now, we add all edges between vertices of color $\alpha$ and $\beta$. Now, recolor all
vertices of color $\alpha$ by the color of $v$. Then, recolor the vertices of color $\beta$ by the
color of $w$.

If we forget a vertex $v$, we first add edges between all vertices of the color of $v$ --- at this 
point, these are all vertices in the clique of $v$, thus effectively ensuring that this set of vertices
indeed is a clique. Then, recolor the vertices with the color of $v$ with the inactive color $o$.
Consider the color of $v$ now unused.

One can verify that this indeed constructs precisely $H$, and that the construction can be done with
$f(k)\log n$ additional space. 
\end{proof}

\subsection{Coloring}\label{sec:col}
In this section we consider the \textsc{Coloring} problem parameterized by linear clique-width.
To prove XNLP-hardness, we reduce from the following problem,
which was shown to be XNLP-complete by Bodlaender et al.~\cite{BodlaenderCW22}.

\defparaproblem%
    {\textsc{Minimum Maximum Outdegree}}
    {Undirected weighted graph $G = (V, E, w)$ with weight function $w\colon E \to \Z^+$ given in unary notation,
        integer $r$, a path decomposition of $G$ of width $\ell$.}
    {$\ell$.}
    {Is there an orientation of $G$ such that for each $v \in V$, the total weight of all edges directed out of $v$ is at most $r$?}
%
%
%
%
\begin{lemma}\label{lem:col:cw:membersip}
    \textsc{Coloring} parameterized by linear clique-width is in XNLP.
\end{lemma}
\begin{proof}
The proof follows the lines of the XP algorithm of \cite{KoblerR03}.
We keep the following certificate in memory: for each nonempty subset $S$ of the $k$ labels, we store the number of colors that appear exactly in this subset of labels. This requires $2^k \log n$ bits.
\begin{itemize}
    \item For operation $\eta_{i \times j}$, we reject if there exists $S \supseteq \{i,j\}$ with at least one color appearing in exactly $S$.
    \item For operation $\rho_{i \to j}$, we simply update the values in our certificate using the inclusion exclusion principle.
    \item For the disjoint union with an isolated vertex $v$ of label $i$, we nondeterministically choose if its color is already present label $i$ or not. In case it is, we do not update the certificate but reject if $i$ was empty (i.e. there is no colour that is exactly present in any subset of labels including $i$). In case it is not, we nondeterministically choose a subset $S$ of labels that does not include $i$ with a corresponding nonzero entry in the certificate, decrement its entry (if it is not the empty set), and then increment the entry of $S \cup \{i\}$.
\end{itemize}

We start from one of the deepest leaves of the expression and apply the previous computations until we reach its root. We then accept if the total over all entries of the certificate is at most the target value.
The computation is done with $O(2^k\log n)$ bits, in FPT time.
\end{proof}


The following gadget will be central in our hardness proof.


%
\begin{lemma}\label{lem:col:cw:gadget}
    For each positive integer $\alpha$, there is a graph $H_\alpha$ with the following properties:
    \begin{enumerate}
        \item\label{cw:gadget:cw}
        There is a linear $10$-expression constructing $H_\alpha$, constructible in linear time and logarithmic space (in $\alpha$),
        that upon completion only uses three labels; with the corresponding vertex sets being $X$, $Y$, and $Z$.
        \item\label{cw:gadget:colorings}
        There is a $2\alpha$-coloring of $H_\alpha$ that uses $\alpha$ colors on $X$ and $2\alpha$ colors on $Y$; 
        and a $2\alpha$-coloring of $H_\alpha$ that uses $2\alpha$ colors on $X$ and $\alpha$ colors on $Y$; we call such colorings \emph{intended}.
        \item\label{cw:gadget:intended}
        Each coloring of $H_\alpha$ with at most $2\alpha$ colors is intended.
    \end{enumerate}
\end{lemma}
\begin{proof}
    The graph $H_\alpha$ is constructed as follows.
    First, we add a clique $Z^* = \{z_1, \ldots, z_{2\alpha}\}$ on $2\alpha$ vertices to $H_\alpha$.
    Then, we add two paths on $2\alpha-1$ vertices, $P^X$ and $P^Y$, and we denote their vertices in the order in which they appear on the path by $p^X_1, \ldots, p^X_{2\alpha-1}$ and by $p^Y_1, \ldots, p^Y_{2\alpha-1}$, respectively.
    We furthermore add two independent sets $S^X = \{s^X_1, \ldots, x^X_\alpha\}$ and $S^Y = \{s^Y_1, \ldots, s^Y_\alpha\}$ on $\alpha$ vertices each, to $H_\alpha$.
    The partition $(X, Y, Z)$ of $V(H_\alpha)$ required by the statement of the lemma will be given by
    $X = S^X \cup \{p^X_i \mid i \bmod 2 = 1\}$, 
    $Y = S^Y \cup \{p^Y_i \mid i \bmod 2 = 1\}$, and
    $Z = V(H_\alpha) \setminus (X \cup Y)$.
    
    For all odd $i \in [2\alpha-1]$, 
    we connect the vertices $p^X_i$ and $p^Y_i$ with an edge.
    To specify additional adjacencies in $H_\alpha$, we define a function $f \colon V(H_\alpha) \setminus Z^* \to 2^{[2\alpha]}$ as follows,
    for all $U \in \{X, Y\}$:
    \begin{align*}
        \mbox{For all odd } i \in [2\alpha-1]: 
            &~f(p^U_i) = \{i', \alpha + i'\}, \mbox{ where } i' = \lceil i/2 \rceil. \\
        \mbox{For all even } i \in [2\alpha-2]: 
            &~f(p^U_i) = \{i', \alpha + i' + 1\}, \mbox{ where } i' = i/2. \\
        \mbox{For all } i \in [\alpha]:
            &~f(s^U_i) = \{i\}.
    \end{align*}
    To finish the construction of $H_\alpha$, for all $v \in V(H_\alpha) \setminus Z^*$, and all $j \in [2\alpha]$, 
    we add the edge $v z_j$ to $H_\alpha$ if and only if $j \notin f(v)$.
    We illustrate this construction in \cref{fig:col:cw:gadget}.
    \begin{figure}
        \centering
        \includegraphics[width=.95\textwidth]{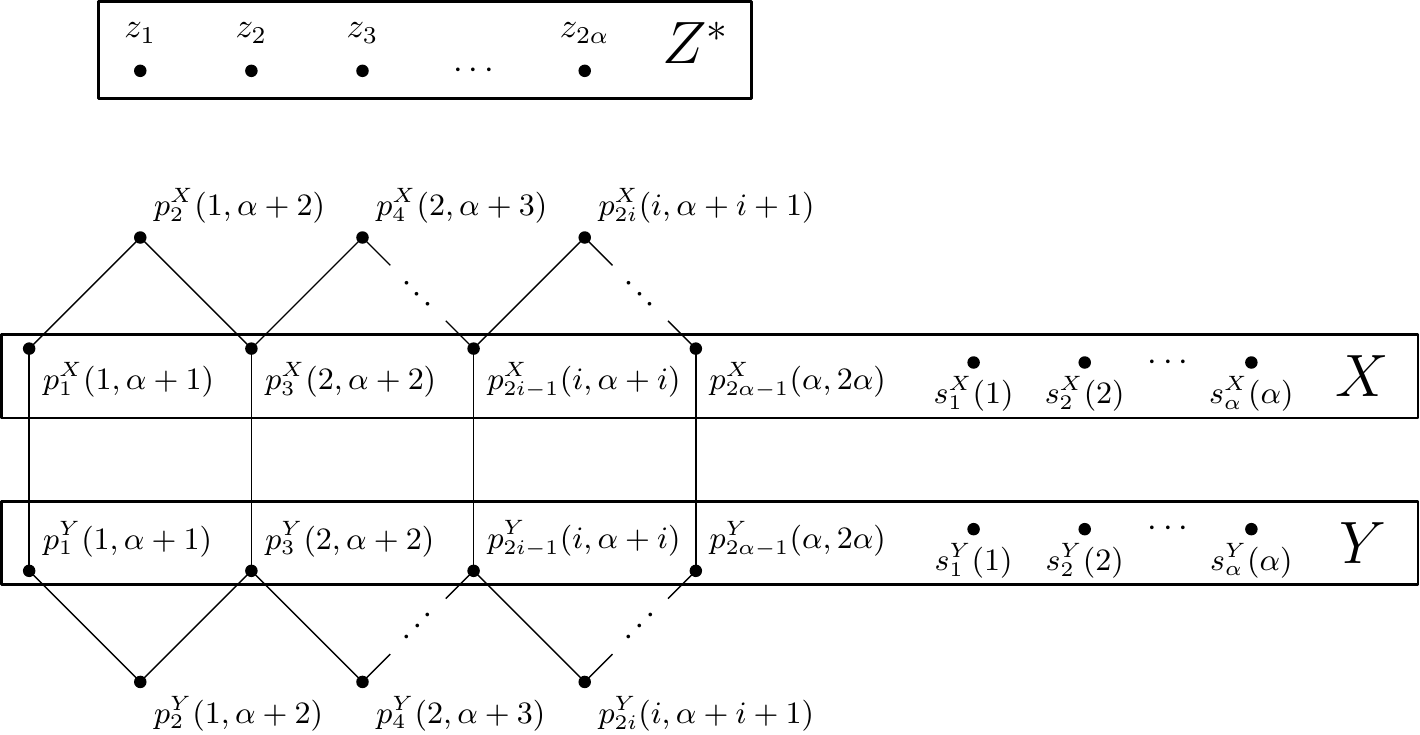}
        \caption{Illustration of the graph $H_\alpha$ from \cref{lem:col:cw:gadget}. All vertices neither in $X$ nor in $Y$ are in $Z$. The set $Z^*$ induces a clique and edges between $Z^*$ and the remainder of the graph are omitted. For all $u \in V(H_\alpha) \setminus Z^*$, there is an edge between $z_j$ and $u$ if $j$ is not on the list shown after vertex $u$.}
        \label{fig:col:cw:gadget}
    \end{figure}
    
    In each proper $2\alpha$-coloring of $H_\alpha$, we may assume that each vertex $z_i$, where $i \in [2\alpha]$, received color $i$.
    We call such $2\alpha$-colorings \emph{canonical}.
    The following observation is immediate from the above construction.
    \begin{observation}\label{obs:cw:canonical}
        In each proper canonical $2\alpha$-coloring of $H_\alpha$,
        each $v \in V(H_\alpha) \setminus Z^*$ receives a color in $f(v)$.
    \end{observation}
    
    We are now able to show that $H_\alpha$ satisfies \cref{cw:gadget:colorings,cw:gadget:intended} of the statement.
    \begin{claim}
        There are precisely two proper canonical $2\alpha$-colorings of $H_\alpha$, 
        one in which colors $[\alpha]$ appear on $X$ and colors $[2\alpha]$ appear on $Y$, and one in which colors $[2\alpha]$ appear on $X$ and colors $[\alpha]$ appear on $Y$.
    \end{claim}
    \begin{claimproof}
        We exhibit the first type of coloring. Throughout, we make use of \cref{obs:cw:canonical} which asserts that the values of the function $f$ can be viewed as lists specifying permissible colors for each vertex not in $Z^*$.
        
        First, we color $p^X_1$ with color $1$, which forces color $\alpha + 1$ on its neighbor $p^X_2$, which forces color $2$ on $p^X_3$. Generally, for odd $i$, $p^X_i$ is forced to receive color $\lceil i/2 \rceil$, and for even $i$, $p^X_i$ is forced to receive color $\alpha + i/2$.
        For each odd $i$, there is an edge from $p^X_i$ to $p^Y_i$, which forces color $\alpha + \lceil i/2 \rceil$ on vertex $p^X_i$.
        This in turn forces color $i/2$ on each vertex $p^Y_i$ for even $i$.
        Lastly, for all $i \in [\alpha]$, 
        the vertices $s^X_i$ and $s^Y_i$ receive their only permissible color, $i$.
        
        The second coloring can be obtained symmetrically, starting by assigning vertex $p^Y_1$ color~$1$.
        Since in each proper canonical $2\alpha$-coloring of $H_\alpha$, either $p^X_1$ or $p^Y_1$ has to receive color $1$, the claim follows.
    \end{claimproof}
    
    It remains to give the linear $10$-expression constructing $H_\alpha$ which satisfies the properties of \cref{cw:gadget:cw}.
    We show a linear $10$-expression constructing a series of graphs $H^i_{\alpha}$, for $i \in [\alpha-1]$, where $H^i_{\alpha}$ is the subgraph of $H_\alpha$ induced on the vertices 
        \[\{z_1, \ldots, z_{i}, 
        z_{\alpha+1}, \ldots z_{\alpha+i}, 
        p^X_1, \ldots, p^X_{2i}, 
        p^Y_1, \ldots, p^Y_{2i}, 
        s^X_1, \ldots, s^X_{i}, 
        s^Y_1, \ldots, s^Y_{i}\}.\]
    
    We use the label set $\{A_0, A_1, A_2, B_0, B_1, B_2, C_0, C_0', C_0'', C_1\}$, and the labeling resulting from the $10$-expression of each $H^i_\alpha$ has the following properties.
    
    \begin{itemize}
        \item The vertices in $V(H^{i}_\alpha) \cap Z^*$ have label $C_0$. 
        \item The vertices in $V(H^{i}_\alpha) \cap Z \setminus V(P^Y) \setminus \{p^X_{2_i}\}$ have label $C_0'$. 
        \item The vertices in $V(H^{i}_\alpha) \cap Z \setminus V(P^X) \setminus \{p^Y_{2i}\}$ have label $C_0''$.
        \item The vertices in $V(H^{i}_\alpha) \cap X$ have label $A_0$.
        \item The vertices in $V(H^{i}_\alpha) \cap Y$ have label $B_0$.
        \item The vertex $p^X_{2i}$ has label $A_2$ and the vertex $p^Y_{2i}$ has label $B_2$.
    \end{itemize}
    
    We now show how how to construct $H^i_\alpha$ from $H^{i-1}_\alpha$ by continuing the linear expression that created $H^{i-1}_\alpha$.
    Note that to construct $H^1_\alpha$, some (parts of some) steps below can be omitted.
    
    \begin{enumerate}
        \item\label{cw:H:alpha:1} 
            Introduce the vertex $p^X_{2i-1}$ with label $A_1$ and $p^Y_{2i-1}$ with label $B_1$.
        \item\label{cw:H:alpha:2}
            Make $A_1$ adjacent to $A_2$ and $B_1$ and make $B_1$ adjacent to $B_2$.
        \item\label{cw:H:alpha:4} 
            Make $A_1$ and $B_1$ adjacent to $C_0$.
        \item\label{cw:H:alpha:5}
            Introduce the vertex $z_{\alpha+i}$ with label $C_1$. 
        \item\label{cw:H:alpha:5:b} 
            Make $C_1$ adjacent to $C_0$, $C_0'$, $C_0''$, $A_0$, and $B_0$.
        \item\label{cw:H:alpha:3} 
            Rename $C_1$ to $C_0$,
            $A_2$ to $C_0'$, and 
            $B_2$ to $C_0''$.
        \item\label{cw:H:alpha:6} 
            Introduce the vertex $z_i$ with label $C_1$. 
        \item\label{cw:H:alpha:6:b}
            Make $C_1$ adjacent to $C_0$, $C_0'$, $C_0''$, $A_0$, and $B_0$.
        \item\label{cw:H:alpha:7} 
            Introduce the vertex $s^X_i$ with label $A_2$ and $s^Y_i$ with label $B_2$.
        \item\label{cw:H:alpha:8} 
            Make $C_0$ adjacent to $A_2$ and $B_2$.
        \item\label{cw:H:alpha:9} 
            Rename $A_2$ to $A_0$ and $B_2$ to $B_0$.
        \item Introduce the vertex $p^X_{2i}$ with label $A_2$ and $p^Y_{2i}$ with label $B_2$.
        \item Make $A_2$ and $B_2$ adjacent to $C_0$.
        \item Make $A_1$ adjacent to $A_2$ and $B_1$ adjacent to $B_2$.
        \item\label{cw:H:alpha:13} 
            Rename $C_1$ to $C_0$, $A_1$ to $A_0$, and $B_1$ to $B_0$.
    \end{enumerate}
    
    To obtain $H_\alpha$ from $H^{\alpha-1}_\alpha$, we first perform \cref{cw:H:alpha:1,cw:H:alpha:2,cw:H:alpha:3,cw:H:alpha:4,cw:H:alpha:5,cw:H:alpha:5:b,cw:H:alpha:6,cw:H:alpha:6:b,cw:H:alpha:7,cw:H:alpha:8,cw:H:alpha:9,cw:H:alpha:13} for $i = \alpha$.
    We finish by relabeling $C_0'$ and $C_0''$ to $C_0$.
    Note that this indeed constructs $H_\alpha$, and that all vertices in $X$ have label $A_0$, all vertices in $Y$ have label $B_0$, and all vertices in $Z$ have label $C_0$.
\end{proof}

We build $H_{\alpha,B}$ from $H_\alpha$ by adding a clique with $B-2\alpha$ vertices, and making these
vertices adjacent to all vertices in $H_\alpha$.

\coloringcw*
\begin{proof}
Membership in XNLP was shown in \cref{lem:col:cw:membersip}.
We prove hardness via a reduction from {\sc Minimum Maximum Outdegree} with pathwidth as parameter.
This problem was shown to be XNLP-complete in \cite{BodlaenderCW22}, and $W[1]$-hard by Szeider~\cite{Szeider11}.


Suppose, we are given a graph $G=(V,E)$ with for each edge a weight $w(e)$, given in unary,
with target value $r$.
Denote for each $v\in V$, the total weight of edges incident to $v$ by
$d(v) = \sum_{\{v,w\}\in E} w(\{v,w\})$. Denote the total weight of all edges by 
$T = \sum_{e\in E} w(e)$.
Build a graph $G'$ as follows. Replace each edge $e$ by the gadget $H_{w(e),B}$.
For an edge $\{v,w\}$, write $V_{v,e}$ for the set $X$ of the gadget and $V_{w,e}$ for the set
$Y$ of the gadget.

To the disjoint union of the gadgets, we add the following edges and vertices. For each vertex
$v$, add edges between vertices in $V_{v,e}$ and $V_{v,e'}$ for all $e\neq e'$, and take
a clique $C_v$ with $2T -d(v) -r$ vertices, and add edges from all vertices
in $C_v$ to all vertices in $V_{v,e}$. Write $V_v$ for the union of all sets $V_{v,e}$ and $C_v$.

\begin{claim}
$G$ has an orientation with maximum weighted outdegree $r$, if and only if $G'$ can be colored with $2T$ colors.
\end{claim}
\begin{claimproof}
Suppose we have an orientation of $G$ with maximum weighted outdegree $r$. 
To each edge $e$, we associate a set $\Gamma_e$ of $2w(e)$ colors, such that these sets are disjoint;
we can do this as the total number of colors is sufficiently large.
Now, for each edge $e=\{v,w\}\in E$, we use the colors of $\Gamma_e$ to color the vertices
in the gadget of $e$. If the edge $e$ is
directed from $v$ to $w$, then we color the gadget in such a way that $2w(e)$ colors are used
for $V_{v,e}$, and $w(e)$ colors are used from $V_{w,e}$; and if
the edge is
directed from $w$ to $v$, then we color the gadget in such a way that $w(e)$ colors are used
for $V_{v,e}$, and $2w(e)$ colors are used from $V_{w,e}$.

For each vertex $v$, the total number of colors used for vertices in sets $V_{v,e}$ over all edges $e$ incident to $v$ equals $d(v)$ plus
the total weight of edges directed out of $v$; by assumption, the latter term is at most $r$.
We now can color each vertex in $C_v$ by a color not used in the sets $V_{v,e}$, as we have
$2T-d(v)-r$ colors left.

\smallskip

Suppose we have a coloring of $G'$ with $2T$ colors. Consider an edge $e=\{v,w\}$. The gadget property
tells that we either use at least $w(e)$ colors for vertices in $V_{v,e}$ and at least $2w(e)$ colors for vertices
in $V_{w,e}$, or
 at least $2w(e)$ colors for vertices in $V_{v,e}$ and at least $w(e)$ colors for vertices
in $V_{w,e}$. In the former case, direct $e$ from $w$ to $v$, and in the latter case, direct
$e$ from $v$ to $w$.

We claim that this orientation has an outdegree that is at most $r$. Consider a vertex $v$.
Note that each of the sets $C_v$, and $V_{v,e}$ for all edges $e$ incident to $v$ uses a different
set of colors. The number of colors used over all sets $V_{v,e}$ is at least 
$d(v)$ plus the total weight of all edges directed out of $v$. For the clique $C_v$, we use $2T-d(v)-r$ colors.
Thus, the total weight of
all edges directed out of $v$ is at most $2T - (2T-d(v)-r) - d(v) = r$.
%
\end{claimproof}

\begin{claim}
$G'$ has linear clique-width at most $k + O(1)$, where $k$ is the pathwidth of $G$. We can construct the corresponding expression using $f(k) + O(\log(n))$ space.
\end{claim}
\begin{claimproof}
Suppose we have a path decomposition of $G$ of width $k$. Transform it to a nice path decomposition, with
introduce vertex, forget, and introduce edge nodes.

We use the following labels: 10 labels for building gadgets, a label $\gamma_j$ for $j\in [1,k+1]$
a label $\delta$ for vertices that will not receive new neighbors.

We visit the bags of the path decomposition in order. For each vertex in the current bag, we have
a unique number in $[1,k+1]$. We describe how to build the expression for $G'$.

Suppose we have an introduce vertex node $X_i$, that introduces vertex $v$. Let $j$ be the smallest
integer in $[1,k+1]$ not assigned to a vertex in $X_{i-1}$, with $j=1$ when $i=1$. Assign $j$ to $v$.
Build the clique $C_v$, giving this clique label $\gamma_j$.

Suppose we have a forget node, that forgets vertex $v$. Suppose $j$ is assigned to $v$. Now,
recolor $\gamma_j$ to $\delta$.

Suppose we have an introduce edge node, that introduces the edge $\{v,w\}$ with weight $\alpha$.
Build the gadget $H_{\alpha,B}$ with the gadget labels. Now, say that the vertices in $X$ have
color $\epsilon_1$, and the vertices in $Y$ have label $\epsilon_2$. Relabel all other gadget
label to $\delta$ (these vertices do not get additional neighbors).
The vertices in $X$, with gadget label $\epsilon_1$ will form $V_{v,e}$, and
the vertices in $Y$, with gadget label $\epsilon_2$ will form $V_{w,e}$. 
Suppose the index of $v$ is $j_1 \in [1,k+1]$, and the index of $w$ is $j_2\in [1,k+1]$.
Add the edges between label class $\epsilon_1$ and label class $\gamma_{j_1}$.
Add the edges between label class $\epsilon_2$ and label class $\gamma_{j_2}$.
Now, relabel $\epsilon_1$ to $\gamma_{j_1}$, and relabel $\epsilon_2$ to $\gamma_{j_2}$.

One can verify that the resulting graph is $G'$. We used $k+O(1)$ labels, so the linear clique-width
of $G'$ is $k+O(1)$.
%
\end{claimproof}

This finishes the proof of the theorem.
\end{proof}
\subsection{Maximum Regular Induced Subgraph}
\label{sec:maxris}

In this section, we consider the \textsc{Maximum Regular Induced Subgraph}
problem and show that this problem is XNLP-complete with linear
clique-width as parameter.

\regularsubgraph*

\begin{proof}
We first show membership in XNLP.
We keep in memory a certificate for each label class corresponding to the number of vertices in the label class that are chosen in the induced subgraph, along with their current degree. 
If the current degree of vertices in the same label class is not uniform then we can already reject, and we therefore only need to store a single degree per label class. 
The choice of a vertex for our induced subgraph is nondeterministic.
We accept if at the end of the linear clique-width expression, the degree of all non empty label classes is the target degree $D$ and the total amount of chosen vertices is at least $W$.

We prove hardness via a reduction from \textsc{Chained Positive CNF-SAT}.
Given an instance $\phi,p,X_1,\dots,X_r,k,(X_{i,j})_{i \in [r],j \in [k]}$ of \textsc{Chained Positive CNF-SAT}, let $C_1,\dots,C_m$ be the clauses in the Boolean formula $\phi$ and let $C[i]$ be the number of clauses that use variables of $X_i$.
We set $D=\max\left\{2\left(1 + \max_{i,j} |X_{i,j}| \right), 2\left\lceil C[i]/2 \right\rceil \right\}$ and choose $N$ to be a large enough integer, e.g. $N=6krD^3 + AD^3$ where $A$ is the total number of literals over all clauses of $\phi$.

We need the following gadgets to create a graph $G$.

\proofsubparagraph{Degree filling gadgets} The following gadget will be used to increase the degree of the variable choice gadgets. Suppose a pair of vertices $v,v'$ has been given. We add five vertices $w_0,w_1,w_2,w_3,w_4$, a $(D-2)$-clique and a $(D-1)$-clique. 
All vertices in the $(D-2)$-clique are adjacent to $w_0, w_1$ and $w_2$. All vertices in the $(D-1)$-clique are adjacent to $w_3$ and $w_4$. Let $w_0$ be adjacent to $v,v'$, $w_1$ be adjacent to both $w_2$ and $w_3$, and $w_2$ be adjacent to $w_4$. See \cref{fig:deg-fill}.

\proofsubparagraph{Variable choice gadget} For each $X_{i,j}$, we add $|X_{i,j}|-1$ pairs of vertices $v,v'$. We add $D$ degree filling gadgets to each pair of vertices, resulting in a  gadget we denote by $\widehat{X_{i,j}}$.

\proofsubparagraph{Clause gadget} For each clause $C$, we add $2N$ vertices $z_1,\dots,z_{2N}$ and $2N$ independent sets $I_1,\dots,I_{2N}$ each of size $\frac{D}{2}-1$. For $i \in [N]$, we make $z_{2i-1}$ and $z_{2i}$ adjacent to each other, and adjacent to $I_{2i-1}$ and $I_{2i}$. The independent sets are placed in a cycle: for $i \in [1,2N-1]$, $I_{i}$ is adjacent to $I_{i+1}$, and $I_{2N}$ is adjacent to $I_{1}$. 

We denote $z_1,z_2$ by $v^C,u^C$. For each $i \in [2,N]$, we add two vertices $y,y'$ and a $(D-1)$-clique, make $y$ and $y'$ adjacent to the clique, and add edges $yz_{2i-1},y'z_{2i}$. See \cref{fig:clause-gadget}. 
Note that all vertices have degree $D$ except for $v^C$ and $u^C$ which have degree $D-1$.

\proofsubparagraph{Degree constraint gadget} For a single vertex $v$, to allow increasing its degree by 2, we add two vertices $a,b$ adjacent to $v$, to each other and to a $(D-1)$-clique. 

To allow increasing the degree of $v$ by $2\ell$ with $\ell>1$, we add $2\ell$ vertices $a_1,\dots,a_{2\ell}$ and $2\ell$ independent sets $I_1,\dots,I_{2\ell}$ each of size $\frac{D}{2}-1$. For $i \in [\ell]$, $a_{2i-1}$ and $a_{2i}$ are adjacent to each other and are both adjacent to $v$, $I_{2i-1}$ and $I_{2i}$. The independent sets are placed in a cycle: $I_1$ is adjacent to $I_{2\ell}$ and for $i \in [1,2\ell-1]$, $I_{i}$ is adjacent to $I_{i+1}$.

\tikzset{ node/.style = {circle,fill,minimum size=5pt, inner sep=0pt, outer sep=0pt}}
\begin{figure}
    \begin{tikzpicture}
    
        \node[node,label=above:$v$] (v) at (0,0) {};
        \node[node,label=above:$v'$] (v') at (0,1) {};
        \node[node,label=above:$w_0$] (w0) at (1,0.5) {};
        \node[circle,draw] (C1) at (2.5,0.5) {$D-2$};
        \node[node,label=below:$w_1$] (w1) at (4,0) {};
        \node[node,label=above:$w_2$] (w2) at (4,1) {};
        \node[node,label=below:$w_3$] (w3) at (5,0) {};
        \node[node,label=above:$w_4$] (w4) at (5,1) {};
        \node[circle,draw] (C2) at (6.5,0.5) {$D-1$};

        \draw (w0) edge (v) edge (v') edge (C1);
        \draw (w1) edge (C1) edge (w3) edge (w2);
        \draw (w2) edge (C1) edge (w4);
        \draw (C2) edge (w3) edge (w4);

    \end{tikzpicture}   
    \centering
    \caption{Degree filling gadget}
    \label{fig:deg-fill}
\end{figure}
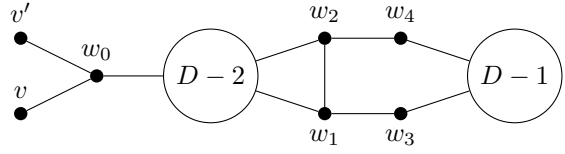

\begin{figure}
    \begin{tikzpicture}
        \node[node,label=above:$v^C$] (vc) at (100:3) {};
        \node[node,label=above:$u^C$] (uc) at (80:3) {};
        \node[circle,draw] (I1) at (105:2) {$I_1$};
        \node[circle,draw] (I2) at (75:2) {$I_2$};

        \node[node] (z3) at (30:3) {};
        \node[node] (z4) at (10:3) {};
        \node[circle,draw] (I3) at (35:2) {$I_1$};
        \node[circle,draw] (I4) at (5:2) {$I_2$};
        \node[node] (y3) at (25:3.5) {};
        \node[node] (y4) at (15:3.5) {};
        \node[circle,draw] (C1) at (20:4.5) {\small $D-1$};

        \node[node] (z2n) at (150:3) {};
        \node[node] (z2n1) at (170:3) {};
        \node[circle,draw] (I2n) at (145:2) {$I_1$};
        \node[circle,draw] (I2n1) at (175:2) {$I_2$};
        \node[node] (y2n) at (155:3.5) {};
        \node[node] (y2n1) at (165:3.5) {};
        \node[circle,draw] (C2) at (160:4.5) {\small $D-1$};

        \node (ld) at (200:2.2) {$\vdots$};
        \node (rd) at (-20:2.2) {$\vdots$};

        \draw (ld) -- (I2n1) -- (I2n) -- (I1) -- (I2) -- (I3) -- (I4) -- (rd);
        \draw (vc) edge (I1) edge (I2) edge (uc);
        \draw (uc) edge (I1) edge (I2);
        \draw (z3) edge (I3) edge (I4) edge (z4);
        \draw (z4) edge (I3) edge (I4);
        \draw (z3) -- (y3) -- (C1) -- (y4) -- (z4);
        \draw (z2n1) edge (I2n1) edge (I2n) edge (z2n);
        \draw (z2n) edge (I2n1) edge (I2n);
        \draw (z2n1) -- (y2n1) -- (C2) -- (y2n) -- (z2n);

    \end{tikzpicture}
    \centering
    \caption{Clause gadget}
    \label{fig:clause-gadget}
\end{figure}

\proofsubparagraph{Variable reading gadget} For each literal $a \in X_{i,j}$ that appears in a clause $C$, we add a vertex $v_a^C$. This is adjacent to the vertices $v_C$ and $u_C$ from the clause gadget and to all pairs of vertices $v,v'$ of the variable gadget $\widehat{X_{i,j}}$. If literal $a$ corresponds to the variable of index $p \in [0,|X_{i,j}|-1]$ in $X_{i,j}$, then we add a degree constraint gadget to allow increasing the degree of $v_a^C$ by $D-2(p+1)$ which is even and at least 2 by definition of $D$. 

\smallskip

Together, the gadgets above constitute graph $G$. We set $W'=(N(2D+1)-D-1)$ (the size of a clause gadget) and set the required size of the $D$-regular induced subgraph that we are looking for to $W=mW'$. We have chosen $N$ sufficiently large such that all clause gadgets need to have at least one vertex included in the subgraph (as $N$ is larger than the size of all non-clause gadgets combined).

\begin{claim}
    $\lcw(G) \leq 2k + \mathcal{O}(1)$
\end{claim}
\begin{claimproof}
    We have an \emph{inactive} label for vertices that are already adjacent to all of their neighbourhood.
    All gadgets can be constructed using a constant amount of fresh labels, and we never construct two gadgets of the same type simultaneously. Each clause only depends on variables from $X_i$ and $X_{i+1}$ for some $i$, from the structure of the considered formula $\phi$. We construct gadgets for increasing $i$. There are $2k$ variable choice gadgets $\widehat{X_{i,j}}$ and $\widehat{X_{i+1,j}}$ (with $j\in [k]$) and we reserve a separate label for each such gadget. We create the variable gadgets, and put all vertices from the degree filling gadgets to inactive (so only the pairs of vertices $v,v'$ keep the label of the gadget). This way, we can keep attaching variable reading gadgets to them, clause by clause.
    Once all clauses containing a literal of $X_{i,j}$ have had their gadget constructed, we can relabel the vertices of $\widehat{X_{i,j}}$ to the inactive label.
\end{claimproof}

\begin{claim}
    If the SAT instance is satisfiable, then there is a $D$-regular induced subgraph of $G$ of size at least $W$.
\end{claim}
\begin{claimproof}
If the variable of index $p$ in $X_{i,j}$ is set to true in the SAT instance, we pick $p$ pairs of vertices $v,v'$ in $\widehat{X_{i,j}}$ into the induced subgraph.
We pick all clause gadgets. For each clause, since it is satisfied, it has a satisfied literal $a$. We pick the vertex $v_a^C$ and its degree constraint gadget.
Now all picked vertices have degree $D$ except for the vertices of the variable choice gadgets.
By definition of $D$, the current degree of the picked vertices in the variable choice gadgets must be less than $D$, so we can complete their degree using the degree filling gadgets. 
From picking all clause gadgets, our induced subgraph has size at least $W$.
\end{claimproof}

\begin{claim}
    If there is a $D$-regular induced subgraph of $G$ of size at least $W$, then the SAT instance is satisfiable.
\end{claim}
\begin{claimproof}
    Let $H$ denote the $D$-regular induced subgraph of $G$ of size at least $W$.
    Since $H$ is of size at least $W$ it must include a vertex of each clause gadget. By design of the clause gadgets, all of their vertices must be included in $H$ for it to be $D$-regular. Furthermore, for each clause $C$, the vertices $v^C,u^C$ must be adjacent to a vertex outside the clause gadget.
    By construction, this vertex can only be a vertex of type $v_a^C$. Since this vertex has degree $D$ in $H$, it must be adjacent to some vertex of the corresponding degree constraint gadget. By design of the degree constraint gadget, all of its vertices must be included in $H$ for it to be $D$-regular. Now since $v_a^C$ has degree $D$, it must be adjacent to exactly $2p$ vertices of the variable choice gadget.
    This means that the variable of index $p$ satisfies the clause $C$, but also that $H$ includes exactly $2p$ vertices of $\widehat{X_{i,j}}$ (not counting the degree filling gadgets).
    
    Now for each $X_{i,j}$, denoting by $r$ the number of vertices $v,v'$ of $\widehat{X_{i,j}}$ that are in $H$, we set the variable of index $\left\lfloor \frac{r}{2} \right\rfloor$ to true.
\end{claimproof}
The two claims above show that we have constructed an equivalent instance of {\sc Maximum Regular Induced Subgraph}. 
It is not hard to see that the transformation can be carried out in
$f(k)\log n$ space, and thus the theorem follows.
\end{proof}
\section{Problems parameterized by pathwidth}\label{sec:rbds}
In this section we consider problems parameterized by pathwidth,
and prove the following theorem.
\cds*
\begin{proof}
We first show membership in XNLP for {\sc Capacitated Red-Blue Dominating Set}. For each red vertex, we guess if it is in the dominating set, and for each edge from
a chosen red vertex to a blue neighbor, we guess if it is used for dominating. We do this while going through the path decomposition from
left to right. We need to keep track which blue vertices are already dominated, which red vertices are in the dominating set plus their remaining capacity, and the total number of vertices in the dominating set so far. We may assume that the remaining capacities are never larger than the number of blue vertices; therefore, we only need to store $O(\log n)$ bits per
vertex in the current bag. Membership in XNLP follows in a similarly for {\sc Capacitated Dominating Set}.

Hardness follows by a reduction from {\sc Circulating Orientation} (defined in Section \ref{sec:maxcut}). 
Suppose that we are given an input of {\sc Circulating Orientation}, say a graph $G=(V,E)$ with weight function $w: E\rightarrow \N$. We assume that these weights are given in unary. Note that in a solution, the total weight of edges directed towards a vertex $v$ and the total weight
of the edges directed out of $v$ should both equal $\sum_{\{v,x\}\in E} w(\{v,x\})/2$.

We build a graph as follows. For each vertex $v\in V$, we create a vertex $v$,  colored red, in $H$. We give $v$ a
private blue neighbor $v'$. The capacity of $v$ equals $1 + \sum_{\{v,x\}\in E} w(\{v,x\})/2$. We can assume this
capacity is integral, otherwise there is no solution to the instance $(G,w)$.

Each edge $e=\{v,x\}\in E$ is replaced by the following gadget. Suppose $w(\{v,x\})=\alpha\in \N$. We create
$2\alpha+ 3$ vertices, called $y_{e,1}, y_{e,2}, \ldots, y_{e,\alpha}$, $z_e$, $z'_e$, $z''_e$, $y'_{e,1}, \ldots, y'_{e,\alpha}$. 
The edge $e$ is replaced by the subgraph shown in Figure~\ref{figure:capredblue}. The vertices $z_e$ and $z''_e$ are red, and all other new vertices are blue. We give the new red vertices $z_e$ and $z'''_e$ a capacity that equals their degree.

\begin{figure}
    \centering
    \includegraphics[scale=2]{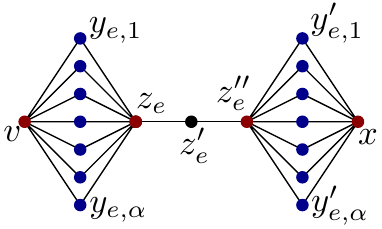}
    \caption{Edge gadget from the proof of Theorem \ref{thm:cds}.}
    \label{figure:capredblue}
\end{figure}

Let $H$ be the resulting red-blue colored graph, with $c(v)$ the capacity of a red
vertex $v\in V(H)$.

We claim that $H$ has a dominating set of size size $|V|+|E|$ for which each chosen red vertex dominates at most its capacity many blue vertices, if and only if $G$ has a circulating orientation.

Suppose first that we have a set $S$ of red vertices with $|S|\leq |V|+|E|$, and an assignment of blue vertices to neighbors in $S$, such that no red vertex has more than its capacity number of vertices assigned to it.

Each vertex that is a copy of a vertex from $V$ must belong to $S$, as they have a private blue neighbor.  For
each edge $e$, either $z_e$ or $z''_e$ must be in $S$, to dominate $z'_e$. This gives in total already $|V|+|E|$ 
vertices, so no edge can have both $z_e$ and $z''_e$ in $S$. 
For each edge $e = \{v,x\}$, if $z_e \in S$, then orient the edge from $v$ to $x$ in $G$; if
$z'''_e\in S$, then orient the edge from $x$ to $v$. 
Now, for each $v\in V$, the total weight of incoming edges of the orientation can be at most $c(v)-1$, since $v$ must also dominate its private neighbor. By definition, $c(v)-1 = \sum_{\{v,x\}\in E} w(\{v,x\})/2$. This means that for each vertex, the total weight of incoming edges is at most half the total
weight of incident edges; it follows that this total weight must be equal, because when there is a vertex
for which this weight is smaller, then there must be another vertex for which it is larger. So, we have an orientation
that is a circulation.

\smallskip 

Suppose now that we have a circulation that is an orientation. Add each original vertex $v\in V$ to $S$, and
for each edge $e=\{v,x\}$, place $z_e$ in $S$ when the edge is oriented from $v$ to $x$ and otherwise place
$z''_e$ in $S$. Red vertices on edge gadgets dominate all their neighbors; red original vertices dominate their private
neighbor and all not yet dominated blue neighbors. This gives a dominating set where each red vertex in $S$ dominates
precisely its capacity many neighbors, as desired.

\smallskip

Finally, we show that we can build a log-space transducer
that transforms a path decomposition of $G$ of width $\ell$ to one of $H$ with width at most $\ell+2$.
We first ensure that the path decomposition of $G$ is nice (which can be done via a log-space transducer). We pass through the bags from left to right. For a forget bag in the path decomposition of $G$, we take the same bag for $H$.
For an introduce bag $X_{i} = X_{i-1} \cup \{v\}$, we loop through the
vertices in $X_{i-1}$ one-by-one, say these are $x_1, \ldots, x_r$. For each $j\in [r]$, if $\{v,x_j\}\in E$, then we add the following bags (in order):
\begin{align*}
&X_i \cup \{z_e, y_{e,1}\},~ X_i \cup \{z_e, y_{e,2}\}, \ldots,
X_i \cup \{z_e, y_{e,w(e)}\},~X_i \cup \{z_e, z'_e\},~X_i \cup \{z'_e, z''_e\},\\
&X_i \cup \{z''_e, y'_{e,1}\},~ X_i \cup \{z''_e, y'_{e,2}\}, \ldots,
~X_i \cup \{z''_e, y'_{e,w(e)}\}.
\end{align*}
One can verify that this gives a path decomposition of $H$. The width has increased by at most $2$.

A standard transformation now shows that \textsc{Capacitated Dominating Set} is also XNLP-hard with pathwidth as parameter.
Given an instance $(G,w)$ of
{\sc Capacitated Red-Blue Dominating Set}, we build an equivalent instance of  {\sc Capacitated Dominating Set}. We give each blue vertex capacity zero. We add two new vertices $x$ and $x'$, with $x'$ of degree one and $x$ adjacent to all red vertices
and to $x'$. The capacity of $x$ is equal to the number of red vertices plus 2. We increase the
target size of the solution by one (and remove all colors). The pathwidth has gone up by at most one.
\end{proof}
We remark that a similar reduction can be used to show XNLP-hardness of \textsc{Capacitated Vertex Cover} (by removing the vertex $z_e'$, having parallel paths of length 3 instead of 2 in the gadget of Figure \ref{figure:capredblue} and giving each original vertex a new neighbor of degree one). 

\section{Problems parameterized by logarithmic linear clique-width}
\label{sec:loglcw}
In this section, we consider problems parameterized by logarithmic linear clique-width 
or logarithmic pathwidth, and prove the following theorem.

\logwidth*

The problem definitions for the parameter logarithmic pathwidth are given below; those for the parameter logarithmic cliquewidth can be obtained by replacing `path decomposition' by `linear clique-width expression'.

\defparaproblem{\textsc{Independent Set (IS)}}{A graph $G=(V,E)$, a path decomposition of $G$ of width $\ell$, and an integer $k$.}{$\left\lceil \ell/\log |V| \right\rceil$}{Is there a subset $S$ of $V$ such that $E \cap S\times S = \varnothing$ and $|S| \geq k$ ?}

\defparaproblem{\textsc{Dominating Set (DS)}}{A graph $G=(V,E)$, a path decomposition of $G$ of width $\ell$, and an integer $k$.}{$\left\lceil \ell/\log |V| \right\rceil$}{Is there a subset $S$ of $V$ such that $N[S]=V$ and $|S| \leq k$?}

\defparaproblem{\textsc{$q$-List-Coloring}}{A graph $G=(V,E)$, a path decomposition of $G$ of width $\ell$, and lists of available colors among $[q]$ for each vertex of $V$.}{$\left\lceil \ell/\log |V| \right\rceil$}{Is there a proper coloring of $G$ which assigns to each vertex of $G$ a color of its list ?}

\defparaproblem{\textsc{Odd Cycle Transversal (OCT)}}{A graph $G=(V,E)$, a path decomposition of $G$ of width $\ell$, and an integer $k$.}{$\left\lceil \ell/\log |V| \right\rceil$}{Is there a subset $S$ of $V$ such that $G-S$ is bipartite and $|S| \leq k$?}

Since $\lcw \leq \pw + 2$, a reduction showing hardness for the parameter (logarithmic) pathwidth also
shows hardness for the parameter (logarithmic) linear clique-width. Conversely, showing membership for the parameter (logarithmic) linear clique-width shows membership for the parameter (logarithmic) pathwidth. Since IS/$\logpw$ and DS/$\logpw$ are already known to be XNLP-complete, we can already conclude that IS/$\loglcw$ and DS/$\loglcw$ are XNLP-hard.
All membership claims in Theorem \ref{thm:logwidth} follow from the following lemma.
\begin{lemma}
    When parameterized by logarithmic linear clique-width, \textsc{Independent Set}, \textsc{Dominating Set}, $q$-\textsc{List-Coloring} for $q>2$, and \textsc{Odd Cycle Transversal} are in XNLP.  
\end{lemma}
\begin{proof}
 The algorithms will keep in memory a certificate of constant size for 
 each of the label classes of the linear clique-width expression. The certificates are as follows.
    \begin{description}
        \item[IS] 
        For each label class, we record whether or not it contains a vertex in the independent set.
        \item[DS] 
        For each label class, we record whether or not it contains a vertex of the dominating set, and whether or not it contains a vertex that is not dominated.
        \item[$q$-\textsc{List-Coloring}] 
        For each label class and for each color, we record whether or not the label class contains a vertex with this color.
        \item[OCT] In addition to constructing an odd cycle transversal, we construct a bipartition of the remaining graph. 
        For each label class, for each side of the bipartition, we record whether or not there is a vertex in both the label class and the side of the bipartition.
    \end{description}
    If there are at most $k$ labels, then the above description uses $O(k\log n)$ space. Membership now follows from applying dynamic programming using the above fingerprints.
\end{proof}

It remains to prove the hardness results, which we do in the next two lemmas.

\begin{lemma}
$q$-\textsc{List-Coloring} and $q$-\textsc{Coloring} parameterized by logarithmic pathwidth are XNLP-hard.
\end{lemma}
\begin{proof}
We reduce from \textsc{Chained Positive CNF-SAT} to \textsc{$q$-List-Coloring} parameterized by logarithmic pathwidth.

Given an instance $\phi,p,X_1,\dots,X_r,k,(X_{i,j})_{i \in [r],j \in [k]}$ of \textsc{Chained Positive CNF-SAT}, suppose that $C_1,\dots, C_m$ are the clauses in $\phi$.

For each $X_{i,j}$ we add dummy variables until $|X_{i,j}|=q^{t_j}$ for some integer $t_j$, increasing its
size by at most a multiplicative factor of $q$ . We enforce that the dummy variables will not appear in any solution by adding a clause containing all initial variables for each $X_{i,j}$.

\proofsubparagraph{Variable choice gadget}
For each $X_{i,j}$, we add vertices $x_{i,j}^1,\dots,x_{i,j}^{t_j}$ with lists $[q]$.
Together they encode the index of a variable in $X_{i,j}$ in base $q$. We denote this gadget by
$\widehat{X_{i,j}}$.

\proofsubparagraph{Clause gadget}
For a clause $C$, with literals $a_1,\dots,a_\ell$, we have a path of length $\ell+2$, with vertices
$p_0,p_1,\dots,p_\ell,p_{\ell+1}$. The vertex $p_0$ is forced to have color $2$, $p_{\ell+1}$ is forced to
have color $2$ if $\ell$ is even, and color $3$ if $\ell$ is odd. Vertices $p_1,\dots,p_\ell$ have
list $[3]$. Note that by construction at least one of them will have to be colored $1$.

\proofsubparagraph{Variable reading gadget}
For a literal $a$ with associated vertex $p$, such that $a \in X_{i,j}$, we read the chosen variable
of $X_{i,j}$ using a \emph{connector} which will allow the color of $p$ to be $1$ only if the
coloring encoding $a$ appears on $(x_{i,j}^1,\dots,x_{i,j}^{t_j})$.
This is done with \emph{color obstruction gadgets}. A \emph{color obstruction gadget} for
color $\alpha \in [q]$ on vertex $x$ can take two forms:
\begin{itemize}
\item If $\alpha = 1$, we have a vertex $w_y$, adjacent to $x$, for each $y \in [q] \setminus \{1\}$,
with list $\{1,y\}$. We add a vertex $w$ adjacent to the $w_y$ and to $p$, with list $[q]$.
\item If $\alpha \neq 1$, we have vertices $w_y$ and $w'_y$ for each $y \in [q]\setminus \{1\}$. We
make $w_y$ adjacent to $x$ and $w'_y$. We give to $w_y$ list $\{\alpha,1\}$, and to $w'_y$ list
$\{y,1\}$. Finally, we add a vertex $w$ adjacent to vertices $p$ and $w'_y$, with full list.
\end{itemize}

The connector of $p$ consists of color obstructions gadgets on each vertex of $\widehat{X_{i,j}}$,
for all colors but the one picked for the encoding of $a$.

\begin{claim}\label{claim:connector-prop}
Consider the connector to a vertex $p$ encoding a variable $z$ of $X_{i,j}$ and a coloring of
$\widehat{X_{i,j}}$.
\begin{enumerate}
    \item Any coloring on $\widehat{X_{i,j}}$ and any color $c \in \{2,3\}$ on $p$ can be extended
    to the connector.
    \item The coloring encoding $z$ on $\widehat{X_{i,j}}$ and color $1$ on $p$ can be
    extended to the connector.
    \item In any coloring of the connector, if $p$ is colored with $1$, then the coloring encoding
    $z$ appears on $\widehat{X_{i,j}}$.
\end{enumerate}
\end{claim}

See \cite[Lemma 11]{LokshtanovMS18known} for proof.

The graph $G$ consists of variable choice gadgets for each $X_{i,j}$, clause gadgets for each clause
of {$\bigwedge \limits_i \phi(X_i,X_{i+1})$}, with each literal vertex having its connector.

\begin{claim}
    If our \textsc{Chained Positive CNF-SAT} instance is satisfiable, then $G$ has a proper list-coloring.
\end{claim}

\begin{claimproof}
    We know how to choose exactly one variable $s_{i,j}$ per $X_{i,j}$ to satisfy the given formula.
    For each $\widehat{X_{i,j}}$, we color it with the encoding of $s_{i,j}$.
    For each clause, we pick an arbitrary satisfied literal and give the vertex representing it
    color $1$. Its connector can be colored by \cref{claim:connector-prop} because the encoding
    on the variable choice gadget must correspond to the satisfied literal.

    The rest of the clause gadget can then be colored greedily. By \cref{claim:connector-prop},
    the coloring can be extended to connectors of vertices not colored with $1$.
\end{claimproof}

\begin{claim}
    If $G$ has a proper list-coloring, the \textsc{Chained Positive CNF-SAT} instance is satisfiable.
\end{claim}

\begin{claimproof}
    We choose set the variables $s_{i,j}$ encoded by the coloring on $\widehat{X_{i,j}}$ to true.
    For each clause gadget, by construction, there must be a vertex with color $1$; denote by $z$
    its associated variable.
    By \cref{claim:connector-prop}, the coloring of its connected variable choice gadget
    $\widehat{X_{i,j}}$ must be the encoding of $z$. Hence the clause is satisfied, since the encoded variable is the one we set to true.
\end{claimproof}
Let $t=\max_j t_j$. We have $t \leq \log_q(n) + 1 = \mathcal{O}(\log(n))$, where $n=|V|$.

\begin{claim}
    $\pw(G) \leq 2kt + 3$.
\end{claim}

\begin{claimproof}
    For each $i \in [r-1]$, the bags corresponding to $\phi(X_i,X_{i+1})$ contain the vertices of
    $\widehat{X_{i,j}}$ and $\widehat{X_{i+1,j}}$ for $j \in [k]$ (at most $2kt$). Then each clause gadget
    together with its connectors can be sweeped with 4 vertices at a time.
\end{claimproof}

To reduce from \textsc{$q$-List-Coloring} to \textsc{$q$-Coloring}, we just add a $q$-clique,
associate each of its vertices to a color, and make them adjacent to all vertices that did not
contain this color in their list. This increases the pathwidth by at most $q$.
\end{proof}

\begin{lemma}
    \textsc{Odd Cycle Transversal} and \textsc{Feedback Vertex Set} parameterized by logarithmic pathwidth are XNLP-hard.    
\end{lemma}

\begin{proof}
We reduce from \textsc{Chained Positive CNF-SAT} to \textsc{Odd Cycle Transversal} parameterized by logarithmic pathwidth.

Given an instance $\phi,p,X_1,\dots,X_r,k,(X_{i,j})_{i \in [r],j \in [k]}$ of \textsc{Chained Positive CNF-SAT}, we denote $\phi=\bigwedge \limits_{i \in [m]} C_i$.

For each $X_{i,j}$, we add dummy variables until $|X_{i,j}|=3^{t_j}$ for some integer $t_j$, increasing its size by at most $3$. They are discarded by adding a clause containing all initial variables for each $X_{i,j}$.

\proofsubparagraph{Variable choice gadget} For each $X_{i,j}$, we add triangles $T_{i,j}^1,\dots,T_{i,j}^{t_j}$, where each triangle has vertices $v_0,v_1,v_2$. The deleted vertices of the triangles will encode the index of a variable in $X_{i,j}$ in base 3. We denote this gadget by $\widehat{X_{i,j}}$.

\proofsubparagraph{Clause gadget} For a clause $C$, with literals $a_1,\dots,a_\ell$, we have a cycle of odd length $\ell'=2\left\lfloor \frac{\ell}{2}\right\rfloor + 1\in \{\ell,\ell+1\}$ with vertices $c_1,\dots,c_{\ell'}$.

\proofsubparagraph{Arrow gadget} An arrow from $u$ to $v$ (see Figure \ref{fig:arrow}) consists of additional vertices $a_1,a_2,a_3,a_4$, $b_1,b_2,b_3$ and edges $ua_1,ub_1,a_1b_1,b_1a_2,b_1b_2,a_2b_2,b_2a_3,b_2b_3,a_3b_3,b_3a_4,b_3v,a_4v$. Note that if $u$ is not in an optimal odd cycle transversal, then $b_1$ and $b_3$ must be. The arrow is then said to be ``passive''. However, if $u$ is in the odd cycle transversal, then we can also have $b_2$ and $v$. The arrow is then said to be ``active''. Intuitively, we got $v$ for free from having $u$.

\begin{figure}
    \centering
    \begin{tikzpicture}
        \foreach \i in {1,...,4} { \node[circle,fill,label=above:$b_\i$] (b\i) at (\i-0.5,1) {};}
        \foreach \i in {1,3} { \node[circle,fill,blue,label=below:$a_\i$] (a\i) at (\i,0) {};}
        \node[circle,fill,label=below:$u$] (a0) at (0,0) {};
        \node[circle,fill,red,label=below:$a_2$] (a2) at (2,0) {};
        \node[circle,fill,label=below:$v$,red] (a4) at (4,0) {};
        \foreach \i in {0,...,3} {
            \draw (a\i) -- (b\number\numexpr\i+1\relax); 
            \draw (a\i) -- (a\number\numexpr\i+1\relax);
            \draw (a\number\numexpr\i+1\relax) -- (b\number\numexpr\i+1\relax);
        }
    \end{tikzpicture}
    \caption{The arrow $A(u,v)$. Vertices in blue form the passive OCT of $A(u,v)$, and those in red form the active OCT of $A(u,v) \setminus \{u\}$.}
    \label{fig:arrow}
\end{figure}
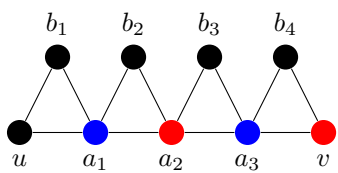

\proofsubparagraph{Negation gadget} A negation from $u$ to $v$ consists of additional vertices $a_1,a_2$ and edges $uv,ua_1,ua_2,va_1,va_2$. Note that we only need to pick one of $u$ and $v$ in our odd cycle transversal to hit the odd cycles of this gadget, and we have to pick at least one.

\proofsubparagraph{Variable reading gadget} Consider a variable $x \in X_{i,j}$ with index $b_1\dots b_{t_j}$, appearing in some clause $C$ as literal $a_z$.
We add vertices $w_1$ and $w_2$, make them adjacent to each other and, for each $a \in [t_j]$, to a vertex $u_a$ which is the endpoint of an arrow from vertex $v_{b_a}$ in triangle $T_{i,j}^a$.
Finally, we have a negation gadget from $w_1$ to the vertex $c_z$ in the clause gadget of $C$.

\smallskip

These gadgets constitute a graph $G$ representing $\phi$. Let $\alpha$ be the number of arrows in our construction, and $\nu$ be the number of negations. We set $W= 2\alpha + \nu + r\sum_{j \in [k]} t_j$.

\begin{claim}
    If the SAT instance can be satisfied then there is an odd cycle transversal of size $W$.
\end{claim}

\begin{claimproof}
    For the triangles in the variable choice gadgets, we pick the vertices corresponding to the encoding of the variable that is assigned. This amounts to $r\sum_{j \in [k]} t_j$ vertices.
    
    For each arrow, we choose the active form when possible. This amounts to $2\alpha$ vertices.
    
    By now, when a variable is assigned, then the variable reading gadgets corresponding to this variable have all of their $u_a$ vertices already picked. Hence, all the cycles of this gadget are hit except for the negation gadget part. We then pick the vertex of the negation gadget that is shared with the clause gadget.
    
    For a variable reading gadget of a variable that was not assigned we can just pick $w_1$ to hit all of the cycles of this gadget.
    
    Regarding the clause gadgets, we know that our assignment satisfies all of the clauses so there should always be at least one literal that was assigned. Therefore, we must have picked a vertex of the clause gadget.
    
    We picked $W$ vertices to form our odd cycle transversal.
\end{claimproof}

\begin{claim}
    If there is an odd cycle transversal of size $W$ then the SAT instance can be satisfied.
\end{claim}

\begin{claimproof}
    Consider an odd cycle transversal $T$ of size $W$.
    To hit the cycles of the variable choice gadgets, we need at least $r\sum_{j \in [k]} t_j$ vertices.
    To hit the arrows, we need at least $2\alpha$ vertices.
    To hit the negation gadgets, we need at least $\nu$ vertices.
    Since this amounts to $W$, we conclude that there is exactly one vertex per triangle in the variable choice gadgets, two vertices per arrow and one vertex per negation gadget.
    
    We consider the assignment corresponding to the encoding stored in the variable choice gadgets. For each $X_{i,j}$, let $b_a$ be the picked vertex of $T_{i,j}^a$. Then in our assignment we set the variable of $X_{i,j}$ with index $b_1\dots b_{t_j}$ to true.
    
    Since $T$ is an odd cycle transversal, we know that for each clause, its corresponding cycle is hit. This means that for one of the negation gadgets, the chosen vertex is the clause gadget vertex. Since the odd cycles of the corresponding variable reading gadget are hit, it must be that all the arrow endpoints were picked, meaning that all these arrows are active, so the vertices that were picked in the variable choice gadget must be encoding the variable.
    We can conclude that our assignment satisfies the clause.
\end{claimproof}

Let $t=\max_j t_j$. We have $t \leq \log_3(n) + 1 = \mathcal{O}(\log(n))$, where $n=|V|$.

\begin{claim}
    $\pw(G) \leq 6kt + 6$.
\end{claim}
\begin{claimproof}
For each $i \in [r-1]$, the bags corresponding to $\phi(X_i,X_{i+1})$ contain the vertices of $\widehat{X_{i,j}}$ and $\widehat{X_{i+1,j}}$ for $j \in [k]$ (at most $6kt$). Then each clause gadget together with its connectors can be swept with 7 vertices at a time: we keep the first vertex of the cycle, the current vertex of the cycle, and the  vertices $w_1$ and $w_2$ of the current variable reading gadget. The 3 remaining vertices suffice to sweep the arrows and the negation gadget.
\end{claimproof}

The instances produced in the previous reduction are also instances of Feedback Vertex Set with similar budget. Indeed, removing an odd cycle transversal from arrows and negation gadgets, disconnects the gadgets and the gadgets do not contain even cycles once their odd cycles are broken. Hence the odd cycle transversals of the construction are also feedback vertex sets. Since a feedback vertex set must hit all cycles, it must hit at least the odd cycles, thus it is also an odd cycle transversal.
\end{proof}
\section{Problems parameterized by linear mim-width}\label{sec:mim}
In this section we show that several fundamental graph problems are XNLP-complete when parameterized by the linear mim-width of the input graph.
For completeness, we state the concrete parameterized problems considered in this section.

\defparaproblem
    {\textsc{Independent Set}}
    {A graph $G$, a linear order of $V(G)$ of mim-width $w$, and an integer $k$.}
    {$w$.}
    {Is there a set $S \subseteq V(G)$ such that $E(G[S]) = \emptyset$ and $|S| \ge k$?}
    
\defparaproblem
    {\textsc{Dominating Set}}
    {A graph $G$, a linear order of $V(G)$ of mim-width $w$, and an integer $k$.}
    {$w$.}
    {Is there a set $S \subseteq V(G)$ such that $S \cup N(S) = V(G)$ and $|S| \le k$?}

\defparaproblem
    {\textsc{Feedback Vertex Set}}
    {A graph $G$, a linear order of $V(G)$ of mim-width $w$, and an integer $k$.}
    {$w$.}
    {Is there a set $S \subseteq V(G)$ such that $G - S$ is a forest and $|S| \le k$?}
    
\defparaproblem
    {\textsc{$q$-Coloring}}
    {A graph $G$, a linear order of $V(G)$ of mim-width $w$, and an integer $k$.}
    {$w$.}
    {Does $G$ have a proper vertex-coloring with $q$ colors?}

XNLP-membership for these problems will be shown via the corresponding dynamic programming XP-algorithms. In all of these algorithms, the following equivalence relation is key to defining the table entries.
\begin{definition}[Neighborhood Equivalence]
    Let $G$ be a graph and $A \subseteq V(G)$.
    For all $X, Y \subseteq A$:
    \(
        X \equiv_A Y \Leftrightarrow N(X) \cap (V(G) \setminus A) = N(Y) \cap (V(G) \setminus A).
    \)
\end{definition}

\begin{lemma}\label{lem:mim:membership}
    The following problems parameterized by the mim-width of a linear order of the vertices of the input graph are in XNLP:
    \begin{enumerate}
        \item\label{lem:mim:membership:IS} \textsc{Independent Set}
        \item\label{lem:mim:membership:DS} \textsc{Dominating Set}
        \item\label{lem:mim:membership:qCol} \textsc{$q$-Coloring} for any fixed $q$.
        \item\label{lem:mim:membership:FVS} \textsc{Feedback Vertex Set}
    \end{enumerate}
\end{lemma}
\begin{proof}
    In all cases, we will show membership using the respective dynamic programming algorithms~\cite{BTV13,JKT20};
    we show here that these algorithms can be implemented using nondeterministic logarithmic space.
    
    Let $v_1, \ldots, v_n$ be the given linear order of the vertices of the input graph, with mim-width $w$.
    With $i$ going from $1$ to $n$, at step $i$ we store partial solutions associated with the subgraph of $G$ induced by the vertices $V_i = \{v_1, \ldots, v_i\}$.
    (For convenience, we let $\overline{V_i} = V(G) \setminus V_i$.)
    In all cases, partial solutions are indexed by constant-size collections of vertex sets 
    each of whose size is bounded by $O(w)$,
    and in some cases, they are representatives of equivalence classes of $\equiv_{V_i}$.
    The following claim has been shown in~\cite{Bui-XuanTV13}, 
    but we reprove it here to clarify that 
    the procedure associated with it can be implemented in logarithmic space.
    \begin{claim}\label{claim:mim:membership:logspace}
        For each $i \in [n-1]$, and each $S_i \subseteq V_i$,
        there is a set $R_i \subseteq V_i$ with $R_i \equiv_{V_i} S_i$ and $\card{R_i} \le w$.
        Furthermore, there is an algorithm using $O(w \log n)$ space
        that determines $R_i$ from $R_{i-1}$,
        where $R_{i-1} \equiv_{V_{i-1}} S_i \cap V_{i-1}$ and $\card{R_{i-1}} \le w$.
    \end{claim}
    \begin{claimproof}
        For $i \le w$, we can simply let $R_i = S_i$,
        so suppose that $i > w \ge 1$,
        and that $\card{S_i} > w$.
        By induction, we can assume that we have $R_{i-1} \subseteq V_{i-1}$ of size at most $w$ such that $R_{i-1} \equiv_{V_{i-1}} S_i \cap V_{i-1}$.
        Let $R_i' = R_{i-1} \cup \{v\}$.
        If $\card{R_i'} \le w$, then we let $R_i = R_i'$ and we are done.
        We may assume that $\card{R_i'} = w + 1$.
        If there is some $x \in R_i'$ such that 
        $N(R_i' \setminus \{x\}) \cap \overline{V_i} = N(R_i') \cap \overline{V_i}$,
        then we let $R_i = R_i' \setminus \{x\}$ and we are done.
        Otherwise, we know that each vertex $x$ in $R_i'$ has a neighbor $y$ in $\overline{V_i}$ such that $y$ 
        is non-adjacent to all vertices in $R_i' \setminus \{x\}$.
        This means that these $xy$-edges form an induced matching in $G[V_i, \overline{V_i}]$, a contradiction.
    \end{claimproof}
    
    The algorithm then works as follows. Upon arrival of the next vertex $v_{i+1}$, we nondeterministically guess its interaction with the solution: in the cases of \textsc{Independent Set}, \textsc{Dominating Set}, and \textsc{Feedback Vertex Set}, whether $v_{i+1}$ is in the solution or not, and in the case of \textsc{$q$-Coloring}, which of the $q$ colors $v_{i+1}$ receives. 
    We then nondeterministically guess the table index corresponding to the updated solution.
    
    In each of the cases, the table entries consist of a collection of a constant number of vertex sets of size
    at most $O(w)$, which implies that the size of each table entry is bounded by $O(w\log n)$.
    For the cases when these sets are representatives of the neighborhood equivalence,
    we can use \cref{claim:mim:membership:logspace} to conclude that the nondeterministic step 
    can be implemented using only $O(w \log n)$ space as well.
    \begin{description}
        \item[\ref{lem:mim:membership:IS}.\ Independent Set~\cite{BTV13}.]
            Here a table index consists of a single representative of an equivalence class of $\equiv_i$, and the table stores the size of a maximum independent set contained in the corresponding equivalence class. 
            Such a table index requires only $O(w\log n)$ bits.
        \item[\ref{lem:mim:membership:DS}.\ Dominating Set~\cite{BTV13}.]
            A table index consists of a pair of equivalence classes, $\mathcal{Q}_i$ of $\equiv_{V_i}$, and $\mathcal{R}_i$ of $\equiv_{\overline{V_i}}$. The table stores the minimum size of a set $Q \in \mathcal{Q}_i$ such that for any $R \in \mathcal{R}_i$, $Q \cup R$ dominates $V_i$. Since we only need two representatives, the table index requires again $O(w\log n)$ bits.
        \item[\ref{lem:mim:membership:qCol}.\ $q$-Coloring for fixed $q$~\cite{BTV13}.]
            Partial solutions are proper colorings of $G[V_i]$ and 
            a table index consists of representatives of equivalence classes $\mathcal{Q}_1, \ldots, \mathcal{Q}_q$ of $\equiv_{V_i}$ such that for all $i \in [q]$,
            color class $i$ in the coloring is contained in $\mathcal{Q}_i$.
            Therefore the table index uses $O(qw\log n)$ bits, which is $O(w\log n)$ since $q$ is a constant.
        \item[\ref{lem:mim:membership:FVS}.\ Feedback Vertex Set~\cite{JKT20}.] 
            The algorithm from~\cite{JKT20} solves the dual problem, \textsc{Maximum Induced Forest}. Here, partial solutions are induced forests $F$ of the subgraph of $G$ induced by $V_i$ and all vertices from $\overline{V_i}$ that have a neighbor in $V_i$.
            The table indices consist of 
            \begin{itemize}
                \item a forest $R$ consisting (roughly speaking) of the restriction of $F$ to $G[V_i, \overline{V_i}]$,
                \item a partition of the connected components of $R$, telling how they are joined together in~$F$,
                \item a minimal vertex cover $S$ of $G[V_i, \overline{V_i}] - V(R)$, and
                \item the number of vertices of $F$.
            \end{itemize}
            The role of $S$ is to control vertices that are either leaves of the solution $F$ or internal vertices of $F$ in $V_i$ that have a neighbor in $\overline{V_i}$. In particular, it indicates that neither of these can have a neighbor in part of the solution coming from $\overline{V_i}$.
            In~\cite{JKT20} it is shown that considering forests on at most $6w$ vertices as candidates for $R$ suffices, which yields that the first and second part of the table index can be represented using $O(w\log n)$ bits.
            The minimal vertex cover $S$ may have many vertices, but it can be replaced by a representative of the equivalence class of $\equiv_{V_i}$ containing $S \cap V_i$, 
            and a representative of the equivalence class of $\equiv_{\overline{V_i}}$ containing $S \cap \overline{V_i}$.
            Storing the size of $F$ clearly only requires $O(\log n)$ bits.
            Therefore, we can have an equivalent definition of the table entries that again only use $O(w\log n)$ bits.
    \end{description}
    This finishes the proof of the lemma.
\end{proof}

The following construction due to Fomin et al.~\cite{FGR20} will be used in the reductions given in this section.
It was used to prove W[1]-hardness of \textsc{Independent Set} and \textsc{Dominating Set} on $H$-graphs, which in turn implied W[1]-hardness of these problems parameterized by linear mim-width.
The \emph{bipartite complement} between two sets $A$ and $B$ in a graph $G=(V,E)$ is obtained by replacing the edges
in $E\cap A\times B$ by the edges in $A\times B \setminus E$.
%
\begin{definition}[\Splitcomp~\cite{FGR20}]\label{def:splitcomp}
    Let $G$ be a graph, and $A, B \subseteq V(G)$ with $A \cap B = \emptyset$.
    The \emph{\splitcomp between $A$ and $B$} is the following operation:
    \begin{enumerate}
        \item Subdivide each edge $uv$ with $u \in A$ and $v \in B$; call the resulting set of vertices $R$.
        \item Take the bipartite complement between $A$ and $R$ and the bipartite complement between $B$ and $R$.
    \end{enumerate}
\end{definition}

The reason why this operation is useful for reductions for problems parameterized by mim-width are the following bounds on the maximum induced matching size of cuts resulting from this construction. This can also be derived from~\cite{FGR20}, but we include a simple direct proof here for completeness.
\begin{lemma}\label{lem:splitcomp}
    Let $G$ be a graph, and $A, B \subseteq V(G)$ with $A \cap B = \emptyset$.
    Let $G'$ be the graph obtained from $G$ by applying the \splitcomp between $A$ and $B$; let $R$ denote the set of vertices created in the construction.
    Then, for all $C \in \{A, B\}$, $\cutmim_{G'}(C, R) \le 2$.
\end{lemma}
\begin{proof}
    Suppose for a contradiction that there is an induced matching of size three in $G'[A, R]$, say $M = \{a_i r_i \mid a_i \in A, r_i \in R\}_{i \in [3]}$. For all $i \in [3]$, let $e_i$ denote the edge in $G$ whose subdivision created vertex $r_i$.
    Since $M$ is an induced matching and by construction, $a_1$ is the endpoint of $e_2$ and $e_3$. But this implies that $a_2$ is not the endpoint of $e_3$, and therefore that the edge $a_2r_3$ exists in $G'$.
\end{proof}

To prove the bound on the mim-width of linear orders constructed in the hardness proofs in this section, we need the following additional lemma
which can be seen as a variation of a lemma in~\cite{BrettelHMP22}, but for linear mim-width.
Recall that for a graph $G$ and a partition $\calP$ of $V(G)$,
the \emph{quotient graph} $G/\calP$ is the graph obtained from $G$ 
by contracting each part of $\calP$ into a single vertex.
The \emph{cutwidth} of a linear order $\Lambda = v_1, \ldots, v_n$, 
denoted by $\cutw(\Lambda)$ is the maximum, over all $i$, 
of the number of edges with one endpoint in $\{v_1, \ldots, v_i\}$ and the other in $\{v_{i+1}, \ldots, v_n\}$.
\begin{lemma}\label{lem:mimw:cutw}
    Let $G$ be a graph, let $\calP = (P_1, \ldots, P_r)$ be a partition of $V(G)$, and let $G' = G/\calP$.
    For all $i \in [r]$ let $\Lambda_i$ be a linear order of $P_i$ such that $\mimw_{G[P_i]}(\Lambda_i) \le c$, and suppose that for all distinct $i, j \in [r]$, $\cutmim_G(P_i, P_j) \le d$.
    Let $\Lambda = \Lambda_1, \Lambda_2, \ldots, \Lambda_r$,
    and let $\Lambda' = P_1, \ldots, P_r$ be the corresponding linear order of $G/\calP$. 
    Then, $\mimw(\Lambda) \le 2d\cdot\cutw(\Lambda') + c$.
\end{lemma}
\begin{proof}
    Let $(A, B)$ be any cut induced by $\Lambda$, and let $M$ be an induced matching in $G[A, B]$. Then, for some $i \in [r]$, 
    \[
    A = P_1 \cup \cdots \cup P_{i-1} \cup (P_i \cap A),
    \mbox{ and } 
    B = (P_i \cap B) \cup P_{i+1} \cup \cdots \cup P_r.
    \]
    The edges of $M$ can now be split into 
    the edges between $A \cap P_i$ and $B \cap P_i$, and 
    for all $h \le i$ and $j \ge i$, 
    where at least one of the inequalities is strict,
    the edges between $P_h$ and $P_j$.
    There are at most $c$ edges of the first kind, 
    since $\mimw_{G[P_i]}(\Lambda_i) \le c$, and $\Lambda$ is equal to $\Lambda_i$ on the vertices in $P_i$.
    For the second kind,
    for each pair $P_h$, $P_j$ when $h < i$ and $j \ge i$,
    we have at most $d$ edges, but only if $P_h P_j$ is an edge in $G/\calP$.
    Therefore the number of such pairs is equal to the size of the cut 
    between positions $i_1$ and $i$ in the 
    linear order $\Lambda'$, and therefore at most $\cutw(\Lambda')$.
    Similarly, the number of pairs
    $P_h$, $P_j$ with $h \le i$ and $j > i$,
    that have edges in $M$
    is equal to the size of the cut between positions $i$ and $i+1$ in $\Lambda'$,
    and therefore again at most $\cutw(\Lambda')$.
    Therefore the total number of the second kind of edges is at most
    $2d\cdot\cutw(\Lambda')$.
\end{proof}

\begin{definition}[Frame graph]\label{def:mim:frame}
    Let $(G, V_1, \ldots, V_r, f)$ be an instance of \CMClique or \CMISet;
    for each $i \in [r]$, let $V(i, 1), \ldots, V(i, k)$ denote the partition of $V_i$ according to $f$.
    
    The \emph{frame graph} $G'$ is obtained from $G$ by
    applying, for each $h \in [r-1]$ and each pair 
    $(i_1, j_1), (i_2, j_2) \in \{h, h+1\} \times [k]$, 
    where $(i_1, j_1) <_{LEX} (i_2, j_2)$,
    the \splitcomp between $V(i_1, j_1)$ and $V(i_2, j_2)$. 
    We denote the set of new vertices by $R(i_1, j_1, i_2, j_2)$.
    
    For convenience, we let $\calP$ denote the partition of $V(G')$ into
    $V(1, 1)$, $\ldots$, $V(r, k)$, $R(1, 1, 1, 1)$, $\ldots$, $R(r, k, r, k)$, and we define the following auxiliary partial function $\phi \colon E(G) \to V(G')$: 
    For all $(i_1, j_1)$ and $(i_2, j_2)$ as above, for each $v_1 \in V(i_1, j_1)$ and $v_2 \in V(i_2, j_2)$ with $v_1v_2 \in E(G)$, we let $\phi(v_1 v_2) \in R(i_1, j_1, i_2, j_2)$ be the vertex created when subdividing $v_1v_2$.
\end{definition}

The following argument will be repeated in several proofs, we therefore extract it as a separate lemma.
\begin{lemma}\label{lem:mim:frame:is}
    Let $(G, V_1, \ldots, V_r, f)$ be an instance of \CMClique,
    and let $G'$ be its frame graph; adapt the notation from \cref{def:mim:frame}.
    Then, $G'$ has an independent set $S$ with $\card{S \cap P} = 1$
    for all $P \in \calP$ if and only if $G$ has a chained multicolored clique.
\end{lemma}
\begin{proof}
    Suppose $G'$ has an independent set $S$ with $\card{S \cap P} = 1$ for all $P \in \calP$.
    Let $v_{i, j} \in S \cap V(i, j)$ for all $i \in [r]$, $j \in [k]$.
    We claim that this implies that 
    for all $h \in [r-1]$, and all $(i_1, j_1), (i_2, j_2) \in \{h, h+1\} \times [k]$ with $(i_1, j_1) <_{LEX} (i_2, j_2)$, 
    we have that $\phi(v_{i_1, j_1} v_{i_2, j_2}) \in S \cap R(i_1, j_1, i_2, j_2)$, 
    which implies that $v_{i_1, j_1}v_{i_2, j_2} \in E(G)$ and in particular that $S \cap V(G)$ is a chained multicolored clique in $G$. 
    Let $r \in S \cap R(i_1, j_1, i_2, j_2)$ and suppose $r \neq \phi(v_{i_1, j_1} v_{i_2, j_2})$.
    We may assume that $r = \phi(v, w)$ where $v \in V(i_1, j_1) \setminus \{v_{i_1, j_1}\}$. 
    But then, $v_{i_1, j_1}r$ is an edge in $G'$, a contradiction.
    
    For the other direction, let $W \subseteq V(G)$ be the chained multicolored clique in $G$.
    Let $S = \emptyset$.
    For each $i \in [r]$ and $j \in [k]$, we add the vertex $v_{i, j} \in W \cap V(i, j)$ to $S$.
    Next, for each $h \in [r-1]$, and each pair $(i_1, j_1), (i_2, j_2) \in \{h,h+1\} \times [k]$ where $(i_1, j_1) <_{LEX} (i_2, j_2)$ (with $<_{LEX}$ the lexicographic ordering), we add $\phi(v_{i_1, j_1} v_{i_2, j_2})$ to $S$.
    Note that since $W$ is a chained multicolored clique, the edge $v_{i_1, j_1}v_{i_1, j_2}$ always exists in $G$.
    It follows immediately from the construction that $S$ is an independent set in $G'$, and that for all $P \in \calP$, $\card{S \cap P} = 1$.
\end{proof}

\begin{lemma}\label{lem:mim:hardness:IS}
    \textsc{Independent Set} parameterized by the mim-width of a given linear order of the input graph is XNLP-hard.
\end{lemma}
\begin{proof}
    We give a parameterized logspace reduction from \textsc{Chained Multicolored Clique} (\CMC);
    let $\calI = (G, V_1, \ldots, V_r, f)$ be an instance of \CMC.
    We create an \ISet instance whose graph $G'$ is the frame graph of $\calI$ (\cref{def:mim:frame}).
    We adapt the notation from \cref{def:mim:frame}.
    We make each $P \in \calP$ a clique in $G'$, 
    and we let $k' = \card{\calP} = (r-1)2k^2 + k$.
    
    Since $G'$ is simply the frame graph of $\calI$ where each $P$ is turned into a clique, correctness of this reduction follows immediately from \cref{lem:mim:frame:is}.
    
    \begin{claim}\label{claim:mim:is:mim}
        There is a logspace-transducer that constructs a linear order $\Lambda$ of $V(G')$ such that $\mimw(\Lambda) = O(k^2)$.
    \end{claim}
    \begin{claimproof}
        For all $i \in [r]$ and $j \in [k]$, we let $\Lambda(i, j)$ be an arbitrary linear order of $V(i, j)$.
        For all $h \in [r-1]$, and all $(i_1, j_1), (i_2, j_2) \in \{h, h+1\} \times [k]$ with $(i_1, j_1) <_{LEX} (i_2, j_2)$, we let $\Gamma(i_1, j_1, i_2, j_2)$ be an arbitrary linear order of $R(i_1, j_1, i_2, j_2)$.
        The desired linear order $\Lambda$ traverses $V(G')$ as follows:
        Consider $(i, j) \in [r] \times [k]$ in lexicographically increasing order.
        First, we follow $\Lambda(i, j)$, 
        and then 
        $\Gamma(i, j, i, j+1)$, $\ldots$, $\Gamma(i, j, i, k)$,
        and if $i < r$, then $\Gamma(i, j, i+1, 1)$, $\ldots$, $\Gamma(i, j, i+1, k)$.
        It is clear that this linear order of $G'$ can be created using $O(\log n)$ bits of memory, where $n$ is the number of vertices of $G$.
        
        Clearly, each $\Lambda(i, j)$ and each $\Gamma(i_1, j_1, i_2, j_2)$ has mim-width at most $1$.
        The only edges in $G'$ between different parts of $\calP$
        are between $V(i_1, j_1)$ and $R(i_1, j_1, i_2, j_2)$, 
        and between $V(i_2, j_2)$ and $R(i_1, j_1, i_2, j_2)$, where $(i_1, j_1) <_{LEX} (i_2, j_2)$.
        By construction it therefore follows from \cref{lem:splitcomp} that for each pair of distinct parts $P_1, P_2 \in \calP$,
        $\cutmim_{G'}(P_1, P_2) \le 2$.
        Let $\Lambda'$ is the linear order of $V(G'/\calP)$ where the vertices of $G/\calP$ appear in the same order as in~$\Lambda$.
        We can observe that $\cutw(\Lambda') = O(k^2)$, 
        and therefore the claim follows from \cref{lem:mimw:cutw}.
    \end{claimproof}
    This concludes the proof of \cref{lem:mim:hardness:IS}.
\end{proof}

\begin{lemma}\label{lem:mim:hardness:DS}
    \DomSet parameterized by the mim-width of a given linear order of the input graph is XNLP-hard.
\end{lemma}
\begin{proof}
    The proof is very similar to that of \cref{lem:mim:hardness:IS}, so we mainly point out the differences;
    the first one is that we reduce from \CMISet instead of \CMClique.
    Let $\calI = (G, V_1, \ldots, V_r, f)$ be an instance of \CMISet.
    We obtain the graph $G''$ of the \DomSet instance in two steps.
    We first create the frame graph $G'$ of $\calI$, 
    adapt the notation from \cref{def:mim:frame},
    and make each $P \in \calP$ a clique in $G'$.
    To obtain $G''$ from $G'$,
    we add a vertex $z_{i, j}$ whose neighborhood is $V(i, j)$, 
    for each $(i, j) \in [r] \times [k]$;
    we let $V(i, j)'' = V(i, j) \cup \{z_{i, j}\}$, and $Z = \{z_{i, j} \mid i \in [r], j \in [k]\}$.
    We let $k'' = rk$.
    
    Suppose that $G$ has a chained multicolored independent set $W$.
    We claim that $W$ is a dominating set in $G''$, and clearly $\card{W} = rk = k''$. Since each $V(i, j)''$ is a clique, and $W \cap V(i, j) \neq \emptyset$, 
    we have that for all $v \in V(i, j)''$, $N_{G'}[v] \cap W \ge 1$.
    Next suppose that there is some $r \in R(i_1, j_1, i_2, j_2)$ with $N(r) \cap W = \emptyset$. Let $e$ denote the edge in $G$ correspding to $r$. The only non-neighbor of $r$ in $V(i_1, j_1)$ is the endpoint $v_1$ of $e$ in $V(i_1, j_1)$, and the only non-neighbor of $r$ in $V(i_2, j_2)$ is the endpoint $v_2$ of $e$ in $V(i_2, j_2)$. This implies that $\{v_1, v_2\} \subseteq W$, 
    but then $e$ is an edge between two vertices in $W$ in $G$, a contradiction.
    
    Now suppose that $G''$ has a dominating set $S$ of size $k''$. Due to the vertices in $Z$, we may assume that each $V(i, j)$ contains a vertex from $S$. Using similar arguments as in the previous paragraph, we can conclude that $S$ is a chained multicolored independent set in $G$.
    Using the same construction as in \cref{claim:mim:is:mim}, only taking into account the vertices in $Z$, we can argue that there is a logspace-transducer constructing a linear order of mim-width $O(k^2)$ of $G''$.
\end{proof}

\begin{lemma}\label{lem:mim:hardness:FVS}
    \FVSet parameterized by the mim-width of a given linear order of the vertices of the input graph is XNLP-hard.
\end{lemma}
\begin{proof}
    Again the proof is very similar to that of \cref{lem:mim:hardness:IS}.
    We give a parameterized logspace reduction from \CMClique
    to the dual problem of \FVSet, \MIForest.
    Let $\calI = (G, V_1, \ldots, V_r, f)$ be an instance of $\CMClique$.
    We first construct $G'$ as the frame graph of $\calI$,
    adapting the notation of \cref{def:mim:frame},
    and making each $P \in \calP$ a clique in $G'$.
    The graph $G''$ of the \MIForest instance is then obtained from $G'$ 
    by adding
    
    \begin{itemize}
        \item for each $P \in \calP$, two vertices $a_P, b_P$ with $N(a_P) = N(b_P) = P$, and
        \item one more vertex $c$ with $N(c) = V(G')$.
    \end{itemize}
    
    We let 
    $P'' = P \cup \{a_P, b_P\}$, 
    $A = \{a_P \mid P \in \calP\}$,
    $B = \{b_P \mid P \in \calP\}$, and 
    $Z = A \cup B \cup \{c\}$.
    We let $k'' = 3\card{\calP} + 1 = 6k^2(r-1) + 3k + 1$.
    
    Suppose that $G$ has a chained multicolored clique $W$.
    Let $S$ be an independent set in $G'$ with $\card{S \cap P} = 1$ for all $P \in \calP$ which exists by \cref{lem:mim:frame:is}.
    We let $S'' = S \cup A \cup B \cup \{c\}$. 
    Since $S$ is also an independent set in $G''$, it follows from the construction that $S''$ induces a tree on $k''$ vertices in $G''$.
    
    Suppose that $G''$ has an induced forest $S$ on $k''$ vertices.
    Since each $P \in \calP$ is a clique, $\card{S \cap P} \le 2$.
    Moreover, the constraint imposed by $k''$ requires us to use three vertices from each $P''$; and if $\card{S \cap P} = 2$, we have that $a_P \cup (S \cap P)$ and $b_P \cup (S \cap P)$ are triangles.
    We therefore have that $\card{S \cap P} = 1$ and that $A \cup B \subseteq S$. Moreover, $k''$ demands that $c$ is in $S$ as well.
    Now, the only way for $S$ to induce a forest in $G''$ is if $S \cap \bigcup_{P \in \calP} P$ is an independent set of size $\card{\calP}$ in $G'$. We can conclude by \cref{lem:mim:frame:is} that this implies a multicolored clique in $G$.
    
    The logspace construction of a linear order of mim-width $O(k^2)$ of $V(G'')$ can once more be done in analogy with \cref{claim:mim:is:mim}.
\end{proof}

\begin{lemma}\label{lem:mim:hardness:qCol}
    For fixed $q \ge 5$,
    \qCol parameterized by the mim-width of a given linear order of the vertices of the input graph is XNLP-hard.
\end{lemma}
\begin{proof}
    We give a parameterized logspace reduction from \CMClique
    to \fiveLCol.
    Let $\calI = (G, V_1, \ldots, V_r, f)$ be the instance of \CMClique.
    We create the graph $G''$ of the \fiveLCol instance as follows:
    Let $G'$ be the frame graph of $\calI$ and adapt the notation of \cref{def:mim:frame}.
    We obtain $G''$ and the lists $L'' = \{L(v) \mid v \in V(G'')\}$ as follows:
    
    \begin{itemize}
        \item For each $P \in \calP$, we add two vertices $a(P)$ and $b(P)$, and make $P'' = P \cup \{a(P), b(P)\}$ a path from $a(P)$ to $b(P)$. We let $\calP'' = \{P'' \mid P \in \calP\}$.
        \item For each $(i, j) \in [r] \times [k]$, each list of a vertex  in $P = V(i, j)$ is $[3]$. If $\card{P}$ is even, the lists of both $a(P)$ and $b(P)$ are $\{1\}$;
        and if $\card{P}$ is odd, the list of $a(P)$ is $\{1\}$, and the list of $b(P)$ is $\{2\}$.
        \item For each $h \in [r-1]$ and $(i_1, j_1), (i_2, j_2) \in \{h, h+1\} \times [k]$ with $(i_1, j_1) <_{LEX} (i_2, j_2)$,
        each list of a vertex  in $R = R(i_1, j_1, i_2, j_2)$ is $\{3,4,5\}$. If $\card{R}$ is even, the lists of both $a(R)$ and $b(R)$ are $\{5\}$;
        and if $\card{R}$ is odd, the list of $a(R)$ is $\{5\}$, and the list of $b(R)$ is $\{4\}$.
    \end{itemize}
    

    The following observation is immediate from the above construction.
    
    \begin{observation}\label{obs:mim:5col:3}
        In each proper list coloring of $(G'', L'')$ and each $P \in \calP$, there is a vertex in $P$ that received color $3$.
        Conversely, if some vertex $v \in P$ received color $3$ in a proper list coloring of $(G'', L'')$, then the vertices in $P'' \setminus \{v\}$ can be properly list-colored with colors $\{1, 2\}$, if $P = V(i, j)$ for some $i, j$ or with colors $\{4, 5\}$, if $P = R(i_1, j_1, i_2, j_2)$, for some $i_1, i_2, j_1, j_2$.
    \end{observation}
    
    Now suppose that $G$ has a chained multicolored clique $W$. Then by \cref{lem:mim:frame:is}, there is an independent set $S$ in $G'$ such that $\card{S \cap P} = 1$ for all $P \in \calP$. 
    Note that $S$ is also an independent set in $G''$.
    We can therefore let $S$ be color class $3$, and by \cref{obs:mim:5col:3}, each path $P$ can be properly list colored without using color $3$. The only remaining edges that need to by checked are between some $V(i_1, j_2)$ and $R(i_1, j_1, i_2, j_2)$, for all valid choices of $i_1, i_2, j_1, j_2$; but since the sets of colors appearing on these two sets of vertices are disjoint, the coloring is proper here as well.
    
    Conversely, suppose that $(G'', L'')$ has a proper list-coloring. Then we can combine \cref{obs:mim:5col:3} and \cref{lem:mim:frame:is} to conclude that $G$ has a chained multicolored clique (with color class $3$ being the independent set required by \cref{lem:mim:frame:is}).
    
    It remains to argue that there is a logspace-transducer creating a linear order $\Lambda''$ of $G''$ of mim-width $O(k^2)$. 
    This can once more be done in analogy with \cref{claim:mim:is:mim}, with the linear order on each $P'' \in \calP''$ now following the path from one endpoint to the other (which has mim-width $1$).
    
    We have shown that \fiveLCol parameterized by the mim-width of a given linear order of the input graph is XNLP-hard. To derive XNLP-hardness of \fiveCol in the same parameterization, observe that we can use the standard trick of adding a clique on vertices $\{1, \ldots, 5\}$, and for each $i \in [5]$, connecting $i$ and $v$ if $i \notin L(v)$. 
    Since adding $c$ vertices can only increase the mim-width of a given linear order by at most $c$, no matter where the new vertices are placed, this does not prohibitively increase the linear mim-width either.
    
    To obtain hardness for any $q > 5$, we simply add $q - 5$ universal vertices to the \fiveCol instance obtained in the previous paragraph. Adding universal vertices cannot increase the mim-width $w$ of any linear order, regardless of where they are placed, unless $w = 0$.
\end{proof}

Combining \cref{lem:mim:membership,lem:mim:hardness:IS,lem:mim:hardness:DS,lem:mim:hardness:FVS,lem:mim:hardness:qCol},
we obtain the main result of this section.
\mimwidth*
\section{Bipartite bandwidth}
\label{sec:bw}
In this section, we show that the \textsc{Bipartite Bandwidth} problem 
is XNLP-complete in the natural parameterization.

A \textit{caterpillar} is a tree where all vertices of degree at least three are on a common path. Bodlaender et al. \cite{XNLP-comp} showed that the following problem is XNLP-complete.

\defparaproblem{\textsc{Bandwidth on Caterpillars}}{A caterpillar $G=(V,E)$ and an integer $k$.}{$k$.}{Is there a bijection $f:V \rightarrow [1,|V|]$ such that for all edges
$\{v,w\}\in E$: $|f(v)-f(w)|\leq k$?}

We reduce from that problem in order to prove the following result.
\bandwidth*
\begin{proof}
We first show membership. Consider a bipartite graph $G=(X,Y,E)$ and integer $k$.
We can assume $|X|=|Y|$ since adding isolated vertices does not increase the bipartite bandwidth.
We will maintain the following certificate: the $2k+1$ last chosen vertices of $X$ and the $2k+1$ last chosen vertices of $Y$. This takes $(2k+1)\log n$ bits.
For $i=1,\dots, |X|$, we guess the $i$th vertex of $x_i$ of in the order $\alpha$ and the $i$th vertex $y_i$ in the order $\beta$. If  $i\geq 2k+1$, we verify that all neighbors of $x_{i-k}$ are among $\{y_{i-2k},\dots,y_i\}$ and that all neighbors of $y_{i-k}$ are among $\{x_{i-2k},\dots,x_i\}$.

We show hardness via a reduction from \textsc{Bandwidth on Caterpillars}. 
Take a caterpillar $G=(V,E)$ and integer $k$.
Set $L = k^4 + 2k^2$. 

Build a graph $H_0$ as follows.
For each $v\in V$, take a vertex $x_v$ in $H_0$ with $2k^4 = 2L- 4k^2$ vertices of degree 1 adjacent to it.
Replace each edge $\{v,w\}\in E$ by a path with $2k$ edges.

Now, take an arbitrary vertex $v\in V$. Take two copies of $H_0$, and add an edge between the two copies of $x_v$, for this vertex $v$. Let $H$ be the resulting graph.

\begin{claim}
If $H$ has bipartite bandwidth at most $L$, then $G$ has bandwidth at most $k$. 
\end{claim}

\begin{claimproof}
Fix a copy of $H_0$. All vertices of the form $x_v$ in this copy are in the same side of bipartition $(X,Y)$ of $H$, say all are in $X$. Suppose $(\alpha,\beta)$ are the orderings of $(X,Y)$ that show that $H$ has bipartite bandwidth at most $L$. 
Let $g$ be the layout of $G$ where $g(v)<g(w)$ iff $\alpha(x_v)<\alpha(x_w)$.
We claim this ordering has bandwidth at most $k$.

Look at $v$, $w$ with $g(w)=g(v)+1$. There are $4k^4$ leaves adjacent to $x_v$ and $x_w$ and for any such leaf $\ell$, we know that $\beta(\ell)\in [\alpha(x_v)-L,\alpha(x_w)+L]$. This implies that $\alpha(x_w)-\alpha(x_v)\geq 4k^4-(2L+1)$.

Consider an edge $\{v,w\}\in E$ and suppose w.l.o.g. that $v$ comes before $w$ in $g$. Repeatedly applying the observation above, if $g(w)=g(v)+q$, then 
\[
\alpha(x_w) \geq \alpha(x_v) + 2q(2k^4-L)-q = \alpha(x_v) + q(2k^4-4k^2-1).
\]
We may assume that $k$ is large enough so that $2k^4-4k^2-1 > L = k^4 + 2k^2$. This implies
that when $g(w)>g(v)+k$, then $\alpha(x_w)>\alpha(x_v)+2kL$,
But there is also a path $x_v=u_0,u_1,\dots,u_{2k}=x_w$ with $2k$ edges between $x_v$ and $x_w$, and we find that \[
|\alpha(u_{2k})-\alpha(u_0)|\leq \sum_{i=0}^{k-1}|\alpha(u_{2i})-\beta(u_{2i+1})|+|\beta(u_{2i+1})-\alpha(u_{2i+2})|\leq 2kL.
\]
This gives a contradiction, so we must have $g(w)\leq g(v)+k$ for all edges $\{v,w\}$. It follows that $G$ has bandwidth at most $k$.
\end{claimproof}

\begin{claim}
If $G$ has bandwidth at most $k$, then $H_0$ has bandwidth at most $L$.
\end{claim}

\begin{claimproof}
We build an ordering on the vertices of $H_0$ as follows. First, we order the vertices
of the form $x_v$ in the same order as the vertices of the form $v$ in the layout of $G$. 
For each $x_v$, insert half of its incident leaves directly before $x_v$ in the ordering, 
and half of its incident leaves directly after $x_v$. 

For each edge $\{v,w\}\in E$, we have to insert the vertices of the path with $2k-1$ vertices
between $x_v$ and $x_w$. Suppose $v$ is earlier in the ordering than $w$, with $y_1, y_2, \ldots, y_\ell$ between $v$ and $w$, in this order, with $\ell\leq k-1$. Write $y_{\ell+1}=w$.
Insert the $2k-1$ vertices on the path consecutively in the ordering as follows: directly
before the first incident leaf of $x_{y_1}$, directly before $x_{y_1}$, directly before the first
incident leaf of $x_{y_2}$, directly before $x_{y_2}$, etc., until we placed a vertex of the path
directly before the first leaf of $x_{y_{\ell+1}}=x_w$. If the path has remaining vertices left, split these evenly before $x_w$ and  before its first leaf.

Let a `slot' be defined as a vertex $x_v$ along with all vertices added before it by the procedure above, or the set of vertices added right before a leaf of a vertex $x_v$.
We now show that all slots have size at most $k(k-1)+1$. 
We first count the vertices added due to an edge $uv$. For each $i \in [k]$, there is at most one such $u$ at distance $i$ from $v$ in the layout, and this contributes $k-i+1$ vertices before the first leaf of $x_v$, and $k-i$ vertices before $x_v$.
We now count the vertices due to edges $uw$ from a vertex $u$ strictly before $v$ to a vertex $w$ strictly after $v$. We count one vertex for each such edge. If $u$ is the $i$th vertex before $v$, for $i \in [k-1]$, then there will be at most $k-i$ such edges coming from $u$.
The number of vertices in the slot of $x_v$ is hence at most 
$\left(\sum_{i=1}^k k-i\right)+\left(\sum_{i=1}^{k-1} k-i \right)=1+k(k-1)$, whereas the number of vertices in the slot before the first leaf of $x_v$ is at most $\left(\sum_{i=1}^k k-i+1\right)+\left(\sum_{i=1}^{k-1} k-i \right)=k^2$. 

If $k>0$, then $1+k(k-1)\leq k^2$, so all slots have at most $k^2$ vertices. Edges of $H_0$ are either on a path between vertices $x_v,x_w$, or edges from $x_v$ to a leaf. The number of vertices in between the endpoints of the edge in the order is hence at most $k^4$ (allowing leaves to be in between) plus $k^2$ (allowing a `slot' to be in between). So we find they are at most $k^4+k^2+1\leq L$ apart. We conclude that $H_0$ has bandwidth at most $L$.

If $k=0$, $G$ has no edge and $H_0$ has no edge so it has bandwidth $L=0$.
\end{claimproof}

\begin{claim}
If $H_0$ has bandwidth at most $L$, then $H$ has bipartite bandwidth at most $L$.
\end{claim}

\begin{claimproof}
Suppose $f$ is a layout of $H_0$ of bandwidth at most $k$. Let $(X,Y)$ be the bipartition of $H_0$. We set $\alpha(x)=f(x)$ for all $x\in X$ and $\beta(y)=f(y)$ for all $y\in Y$. For a copy $x'$ of the vertex $x\in X$ in the second copy of $H$, we set $\beta(x')=f(x)$, and similarly we set $\alpha(y')=f(y)$ for a copy $y'$ of a vertex $y$. We find that $|\alpha(x)-\beta(y)|=|f(x)-f(y)|\leq L$ whenever $\{x,y\}$ is an edge of $H_0$, and similarly $|\alpha(y')-\beta(x')|\leq L$. There is only one edge $\{x,x'\}$ between vertices of the two copies, and $\alpha(x)=\beta(x')$.
\end{claimproof}

The three claims show that $H$ has bipartite bandwidth at most $L$ if and only if $G$ has bandwidth at most $k$, as desired.

\end{proof}

\section{Conclusion}
\label{section:conclusion}
In this paper, we gave a number of XNLP-completeness proofs for graph problems 
parameterized by linear width measures. 
Such results are interesting for a number of reasons: 
they pinpoint the ``right'' complexity class for parameterized problems, 
they imply hardness for all classes $W[t]$, 
and they tell that it is unlikely that there is an algorithm that uses
`XP time and FPT space' by \cref{conj}.


This paper gives among others
the first examples of XNLP-complete problems when the
linear clique-width or linear mim-width is taken as parameter.
Our hardness results give new starting points for future hardness and
completeness proofs,
in particular for problems with width measures like pathwidth, (linear) clique-width, or
(linear) mim-width as parameter. 

Other interesting directions for future research in this line are, for instance,
to consider the parameterization by \emph{cutwidth}; 
a promising candidate problem to show XNLP-completeness parametized by cutwidth is \textsc{List Edge Coloring}.
Another interesting parameter to consider in this context is the \emph{degeneracy} of a graph.
It is also interesting to explore the concept of XNLP-completeness for width measures of other objects than graphs.
One could for instance consider (linear) width measures of digraphs or hypergraphs.

We also leave open what the correct parameterized complexity class is for {\sc Feedback Vertex Set} parameterized by logarithmic pathwidth or logaritmic linear cliquewidth -- we showed XNLP-hardness, but did not prove containment in XNLP.

\bibliographystyle{plainurl}
\bibliography{main}


\end{document}